\documentclass[10pt,journal,compsoc]{IEEEtran}
\usepackage{amsmath}

\usepackage{url}

\usepackage{graphicx}

\usepackage{enumitem}
\usepackage{pifont}
\newcommand{\cmark}{\ding{51}}%
\newcommand{\xmark}{\ding{55}}%
\usepackage{multirow}
\usepackage{amsmath}
\usepackage{cancel}
\usepackage{amssymb}
\usepackage{amsfonts}
\usepackage{array}
\usepackage{tcolorbox,cuted,lipsum}
\usepackage{tikz}
\usepackage[T1]{fontenc}
\usepackage{subcaption}
\usepackage{algorithmicx,float}

\newtheorem{claim}{\textbf{Claim}}

\newtheorem{theorem}{\textbf{Theorem}}
\newtheorem{definition}{\textbf{Definition}}
\newtheorem{proof}{\textbf{Proof}}

\usepackage[linesnumbered,ruled]{algorithm2e}
\newenvironment{proc}[1][htb]
  {
   \begin{algorithm}%
  }{\end{algorithm}}
\newenvironment{module}[1][htb]
  {
   \begin{algorithm}%
  }{\end{algorithm}}
  
\usepackage{hyperref}
\usepackage{float}

\ifCLASSOPTIONcompsoc
  \usepackage[nocompress]{cite}
\else
  \usepackage{cite}
\fi
\ifCLASSINFOpdf
\else
\fi
\begin{document}
%
\title{CryptoMaze: Privacy-Preserving Splitting of Off-Chain Payments }
%
%
%
%

\author{Subhra~Mazumdar,~and~Sushmita~Ruj,~\IEEEmembership{Senior~Member,~IEEE} \\
\IEEEcompsocitemizethanks{\IEEEcompsocthanksitem Subhra Mazumdar is with the Cryptology and Security Research Unit, Indian Statistical Institute, Kolkata, India  \protect\\
E-mail: subhra.mazumdar1993@gmail.com
\IEEEcompsocthanksitem Dr. Sushmita Ruj is with University of New South Wales, Sydney, Australia  \protect\\
E-mail: Sushmita.ruj@unsw.edu.au

}
}

%
%

\markboth{}%
{Shell \MakeLowercase{\textit{et al.}}: Bare Demo of IEEEtran.cls for Computer Society Journals}
%



\IEEEtitleabstractindextext{%
\begin{abstract}

Payment Channel Networks or PCNs solve the problem of scalability in Blockchain by executing payments off-chain. Due to a lack of sufficient capacity in the network, high-valued payments are split and routed via multiple paths. Existing multi-path payment protocols either fail to achieve atomicity or are susceptible to wormhole attack. We propose a secure and privacy-preserving atomic multi-path payment protocol CryptoMaze. Our protocol avoids the formation of multiple off-chain contracts on edges shared by the paths routing partial payments. It also guarantees unlinkability between partial payments. We provide a formal definition of the protocol in the Universal Composability framework and analyze the security. We implement CryptoMaze on several instances of Lightning Network and simulated networks. Our protocol requires 11s for routing a payment of 0.04 BTC on a network instance comprising 25600 nodes. The communication cost is less than 1MB in the worst-case. On comparing the performance of CryptoMaze with several state-of-the-art payment protocols, we observed that our protocol outperforms the rest in terms of computational cost and has a feasible communication overhead.

\end{abstract}

}

\maketitle

\IEEEdisplaynontitleabstractindextext

%
\IEEEpeerreviewmaketitle

\IEEEraisesectionheading{\section{Introduction}\label{sec:introduction}}

Cryptocurrencies are gaining prominence as an alternative method of payment. Blockchain forms the backbone of such currencies, guaranteeing security and privacy. It allows transacting parties to remain pseudonymous and ensures the immutability of records. Records in blockchain are publicly verifiable. Bitcoin mining relies on Proof-of-Work (PoW) \cite{nakamoto2008bitcoin, o2014bitcoin, bano2019sok} to ensure a Sybil-resistant network. Unfortunately, PoW is quite resource-intensive and time-consuming, reducing transaction throughput \cite{croman2016scaling, poon2016bitcoin}. 

Layer-two protocols provide a solution to the problem of scalability. It enables users to perform transactions \emph{off-chain} and massively cut down data processing on the blockchain. Solutions like payment channels, channel factories, payment channel hub, side-chains, and commit-chains have been stated in the literature survey \cite{gudgeon2020sok}. \textit{Payment Channels} \cite{decker2015fast}, \cite{poon2016bitcoin} are widely deployed in many applications. It is modular and does not require any fundamental changes in the protocol layer. Two parties can mutually agree to open a payment channel by locking their funds for a certain period. Nodes not directly connected by a payment channel route a payment via an existing set of channels. This set of interconnected payment channels forms a \emph{Payment Channel Network} or PCN. \emph{Lightning Network} for Bitcoin \cite{poon2016bitcoin} and \emph{Raiden Network} for Ethereum \cite{raiden} are the two most popular networks. Designing privacy-preserving routing and payment protocols for such networks is a big challenge. Most of the routing algorithms focus on finding a single path for routing a transaction. However, finding a single route for a high-valued transaction is a challenging task. After several payments get executed in the network, channels in a path may not have sufficient balance to relay the funds. In such circumstances, it is better to split high-valued payments across multiple paths to increase the success rate of transactions. However, it is not trivial to design a protocol for multi-path payment and we discuss the challenges faced.

%

%
%
%

\subsubsection*{Challenges faced in multi-path payments}
\begin{itemize}[leftmargin=*]
\item \emph{Atomicity of payments}: Several distributed routing algorithms \cite{elias,prihodko2016flare,silentwhispers,speedymurmur,viswanath2012canal,yu2018coinexpress,hoenisch2018aodv,wang2019flash,mazumdar2020hushrelay,lin2020rapido} have been proposed for relaying transactions across multiple paths. A payment transferred from payer to payee must be \emph{atomic}. Either all the partial payments succeed or fail in their entirety. Applying existing payment protocols like Hashed Timelock Contract \cite{poon2016bitcoin}, \cite{moreno2015privacy}, \texttt{BOLT} \cite{green2017bolt}, Sprites \cite{miller2017sprites}, \cite{malavolta}, Anonymous Multi-Hop Lock or \texttt{AMHL} \cite{malavoltamulti}, on individual paths routing partial payment might not guarantee atomicity. If an instance of the protocol fails in one of the paths, only the partial amount gets transferred to the receiver, violating atomicity.
 
\item \textit{Susceptible to wormhole attack}: Existing multi-path payment protocols like  \emph{AMP}\cite{multipath}, \emph{Boomerang} \cite{bagaria2019boomerang} achieve atomicity. Each path forwarding the partial payment uses the same commitment, making it susceptible to \emph{wormhole attack} \cite{malavoltamulti}. Malicious parties in a given path may collude and steal an honest party's processing fee.

\item \textit{Multiple off-chain contracts on shared channels}: Multiple paths routing a single payment may not be edge-disjoint. In~\autoref{multiple}, $M$ wants to transfer $5.1 \ units$ to $N$. The payment is split across two paths $p_1=\langle MA\rightarrow AB \rightarrow BD \rightarrow DN\rangle$ and $p_2=\langle MA\rightarrow AC \rightarrow CD \rightarrow DN\rangle$ into 2.6 units and 2.5 units respectively. Each intermediate parties charge a processing fee of 0.1 units. Channels $MA$ and $DN$ are shared by the two paths. Thus, two off-chain contracts need to be established for routing each partial payment. Also, nodes $A$ and $D$ get paid twice for forwarding each partial payment, levying an additional cost overhead on the sender $M$. To save cost and avoid the overhead of instantiating off-chain contracts, it is better to construct one off-chain contract on shared payment channels for a payment instance.

 \begin{figure}[!ht]
    \centering
    \includegraphics[scale=0.35]{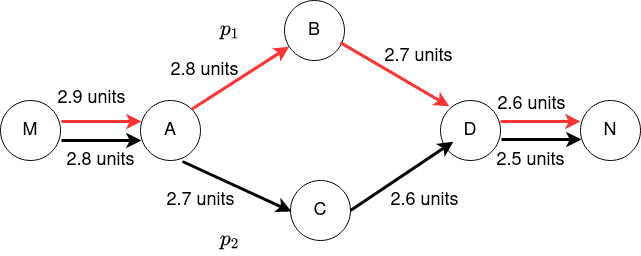}
    \caption{Paths $p_1$ and $p_2$ shares channels $MA$ and $DN$}
    \label{multiple}
\end{figure}

\item \textit{Linkability between partial payments}: A given node will be willing to route full payment instead of partial payments \cite{eckey2020splitting}. The success rate is low when payment is split. If a partial payment fails in one of the paths, then the entire payment rolls back. If colluding parties can link partial payments, they will tend to reject such requests and preserve their channel capacity for routing the full amount. Unlinkability must be ensured to prevent censoring split payments.

\end{itemize}
Our goal is to construct a payment protocol that addresses all the shortcomings discussed above. 

\vspace*{-0.1cm}

\subsection{Contributions}

\begin{itemize}[leftmargin=*]

\item We propose CryptoMaze, an efficient, privacy-preserving, atomic multi-path payment protocol. Our protocol optimizes the setup cost by avoiding the formation of multiple off-chain contracts on a channel shared by partial payments. To date, no other protocol has been able to achieve this optimization.

\item Our protocol ensures balance security, i.e., honest intermediaries do not lose coins while forwarding the payment. 
\item Our protocol description ensures unlinkability between partial payments.

\item We have modeled CryptoMaze and defined its security and privacy notions in the Universal Composability or UC framework.

\item Experimental Analysis on several instances of Lightning Network and simulated networks show that our proposed payment is as fast as Atomic Multi-path Payment \cite{multipath}. The run time is around 11s for routing a payment of 0.04 BTC in a network instance of 25600 nodes. The communication overhead is within feasible bounds, being less than 1MB. The code is available in \cite{Code}.
\end{itemize}

\subsection{Organization}
The rest of the paper has been organized as follows: Section \ref{3} provides the background concept needed for understanding our paper. Section \ref{rel} discusses the related works in multi-path payments. Our proposed protocol has been described in detail in Section \ref{5}. We discuss the security of our protocol in Universal Composability (UC) framework in Section \ref{attacksec} and provide the security analysis. The experimental observation has been provided in Section \ref{exp}. Finally, we conclude our paper in Section \ref{7}.


\section{Background }
\label{3}
 \begin{table}
\centering
  

  \begin{tabular}{|c |c |} 
    \hline
 Notation &Description \\
 \hline
 
$\mathcal{G}:=(V,E)$ &Bidirected Graph representing PCN\\
$V$   &Nodes in $\mathcal{G}$\\
$E$   & Payment channels in $\mathcal{G}$, $E \subset V \times V$\\
$C: E \times \mathbb{N} \rightarrow \mathbb{R}^{+}$ & Capacity function\\
 $f: V \rightarrow\mathbb{R}^{+}$ &Function defining processing fee\\
 $id_{i,j}$ & Identifier of payment channel $(U_i,U_j)\in E$\\
 $\mathbb{B}$ &Blockchain \\

 $\mathbb{G}$ &Elliptic curve of order $q$\\
 & where $q=p^n, p$ is a prime number\\
 $G$ &Base point of elliptic curve $\mathbb{G}$\\
 $\lambda$ &Security Parameter\\
 $\mathcal{H}:\{0,1\}^*\rightarrow \{0,1\}^\lambda$ & Standard Cryptographic Hash function\\
 $\Delta$ &Worst-case confirmation time for a\\
 &  transaction to get recorded in Blockchain\\
 $U_0$ & Payer, a node in set $V$ \\
  $U_r$ & Payee, a node in set $V$\\
  $gain  : V \rightarrow \mathbb{R}^{+}$ & Function defining coins gained by a node\\
  $\mathbb{PC}$ &Set of payment channels, created by $U_0$\\

  $T_i$ &Timestamp at which  node $U_i$ receives\\
  &  its first incoming contract request. \\
    $\delta$ &latency\\
  $\mathcal{F}$ & Ideal functionality for payment in PCN\\
    $\mathcal{F}_{\mathbb{B}}$ & Ideal functionality for Blockchain $\mathbb{B}$\\
        $\mathcal{F}_{smt}$ & Ideal functionality for secure message\\
        & transmission\\
  \emph{Sim} &Ideal world adversary\\
  $\mathcal{A}$ &Adversary in the real world\\
  $\mathcal{Z}$ &Environment\\
  
\hline

\end{tabular}

\caption{Notations used in the paper}
\label{tab:not}
\vspace{-0.2cm}
\end{table}

In this section, we provide the required background for understanding our protocol. The terms source/payer means the sender node. Similarly, sink/payee/destination means the receiver node. A payment channel has been referred to as an edge. \autoref{tab:not} states the notations used in the paper.


\subsection{Payment channels}
\label{pc}

A payment channel enables several payments between two users without committing every single transaction to the blockchain.  Any two users can mutually agree to open a payment channel by locking their coins into a multi-signature address controlled by both users. These parties can perform several off-chain payments by locally agreeing on the new deposit balance. Correctness of payments is enforced cryptographically by the use of hash locks, time locks \cite{poon2016bitcoin}, or scriptless locking \cite{malavoltamulti}. A party can close the payment channel, with or without the cooperation of counterparties, broadcasting the latest transaction on the Blockchain. Broadcasting of older transactions leads to the slashing of deposits made by the malicious party. 

\subsection{Payment channel networks (PCNs)}
\label{basic}
A Payment Channel Network is modeled as a bidirected graph $\mathcal{G}=(V,E)$ where $V$ represents the participants in the network and $E \subseteq V \times V$ denotes the payment channels existing between parties \cite{poon2016bitcoin}, \cite{decker2015fast}. Opening a payment channel $(U_i,U_j)$ is equivalent to the opening of two unidirectional payment channels $(U_i,U_j)$ and $(U_j,U_i)$. The channel identifier for $(U_i,U_j)$ is denoted as $id_{i,j}$. The underlying blockchain, denoted as $\mathbb{B}$, acts like a trusted append-only ledger recording the opening and closing of payment channels. 
A capacity function, defined as $C: E \times \mathbb{ N}\rightarrow \mathbb{R}^+$, denotes the balance of each party in the channel at a given time. For example, $C((U_i,U_j),t)$ denotes the balance of party $U_i$ in the channel $id_{i,j}$ at time $t$.  We define the fee charged by each node as $f: V \rightarrow \mathbb{R}^{+}$. The fee calculated is proportional to the coins a particular node is routing through its channel. If a party $U_i$ receives a request to transfer $val$ coins at time $t_{current}$ to a node $U_j$, it checks locally whether there exist payment channels connected to $U_i$ and $C((U_i,U_j),t_{current}) \geq val$.
\subsection{Off-chain contracts}

Off-Chain contracts are smart contracts where the logic encoded in the contract is not run by the miners. It is mutually executed by the participants involved in instantiating the contract. The advantage of having off-chain contracts are that computation-intensive tasks can be executed without involving blockchain as long as participants behave honestly. An individual player can prove the correct contract state independently. Cheating is prevented as the state of the contract is signed by all the players. If a party misbehaves by broadcasting a wrong state in blockchain, the counterparty can raise a dispute and publish the valid accepted state. Hashed Timelock Contract or HTLC \cite{poon2016bitcoin} is one such example used in PCN for routing payments in the network. The logic used is a hash function, where players need to provide the preimage of the hash to claim coins.
\subsection{Wormhole attack in PCNs}
HTLC uses the same commitment across the path routing the payment. Consider an example where $U_0$ wants to transfer $\alpha$ coins to $U_r$ via nodes $U_1,U_2,\ldots,U_{n-1}, n>3$. The coins transferred by $U_0$ is $\alpha+\sum\limits_{i=1}^{n-1}f(U_i)$. Node $U_3$ colludes with $U_{n-2}$ before the protocol starts. In the release phase, $U_r$ decommits and claims the coins. However, node $U_{n-2}$ directly shares the decommitment with $U_3$. The former cancels the HTLC with node $U_{n-3}$ and the cancellation of HTLC continues till node $U_3$. So nodes from $U_4$ to $U_{n-3}$ considers the payment to have failed. $U_3$ and $U_{n-2}$ steals the fee of all these intermediate nodes, gains $\sum\limits_{i=4}^{n-3}f(U_i)$. This is termed wormhole attack \cite{malavoltamulti}.

\section{Related works}
\label{rel}
\begin{table}[ht]
    \centering
    \begin{tabular}{|c|c|c|c|c|c|c|}
    \hline
          &AMP \cite{multipath} &\cite{piatkivskyi2018split}   &\cite{bagaria2019boomerang} & NAPS \cite{dziembowski2020non} &\cite{eckey2020splitting}  &Crypto-\\
          & & & & & &Maze\\
    \hline
	At  &\cmark  &\xmark &\cmark &\xmark &\cmark &\cmark\\
			WA &\cmark &\cmark &\cmark &\xmark &\xmark &\xmark\\
			Li &\xmark &\xmark &\xmark &\cmark &\cmark &\xmark\\
          M-OC &\cmark &\cmark &\cmark &\cmark &\cmark &\xmark\\
    \hline
    \end{tabular}
    \caption{Comparative Analysis of CryptoMaze with existing Multi-path payment protocols in terms of atomicity (At), wormhole attack (WA), Linkability (Li) and multiple off-chain contracts on shared edges (M-OC)}
    \label{tab:my_label}
\end{table}

Several single path payment protocols like Hashed Time-Lock Contract or \emph{HTLC} \cite{poon2016bitcoin}, Multi-Hop HTLC \cite{malavolta}, Anonymous Multi-Hop Lock or \emph{AMHL} \cite{malavoltamulti} have been proposed that works for single-path payment. However, a direct extension of such protocols into multi-path payment may fail to guarantee atomicity. Sprites \cite{miller2017sprites} was proposed for Ethereum-styled PCN guarantees atomicity of payments and locks constant collateral. In \cite{eggeratomic}, \cite{aumayr2021blitz}, a similar construction has been proposed Bitcoin-compatible PCN. However, such protocols work for single-path and lack any discussion on multi-path settings.

Multi-path payment was first discussed in SilentWhisper \cite{silentwhispers}, but at the cost of substantial computation overhead. The protocol was not atomic. Osuntokun \cite{multipath} was the first to propose a protocol that guarantees the atomicity of split payments. It uses linear secret sharing of the commitments shared across the multiple paths routing partial payments. But this protocol is susceptible to \emph{wormhole attack} and high latency. In \cite{piatkivskyi2018split}, a protocol for \emph{splitting payments interdimensionally} was proposed where the total amount to be transferred is split into unit-amounts and routed through the same or different routes. However, the authors state that their protocol does not stress achieving atomicity. Partial satisfaction of payment is considered a favorable outcome. The problem of latency and throughput in \emph{AMP} is addressed by another payment protocol, Boomerang \cite{bagaria2019boomerang}. However, the protocol suffers from the problem of \emph{wormhole attack} and requires locking of excess collateral. In \cite{lin2020rapido}, a payment protocol termed D-HTLC was proposed for multiple paths. However, the protocol relies on the atomicity of payments using a penalization mechanism. Levying penalty is not a good method since honest nodes might lose coins without any fault. A protocol \emph{Non-Atomic Payment Splitting (NAPS)} that recursively splits payment is discussed in \cite{dziembowski2020non}. However, the protocol does not aim for atomicity and partial payment is treated as a valid outcome. Eckey et al. \cite{eckey2020splitting} had proposed an atomic payment protocol that allows intermediaries to split payments dynamically by adapting to the local condition. The protocol is atomic, privacy-preserving, and not susceptible to wormhole attack. However, each node forwarding payment uses homomorphic encryption to encrypt the payment information. Such an operation is quite computation-intensive. The public key of the receiver is forwarded to all the nodes routing partial payments. Though the authors claim that partial payments remain unlinkable, colluding parties can link payments by observing the common public key.

We provide a comparative analysis of our protocol with the state-of-the-art multi-path payment in \autoref{tab:my_label}. Our protocol is atomic, wormhole attack resistant, and guarantees unlinkability between partial payments. None of the shared edges require multiple off-chain contracts for a single payment instance. A new protocol, xLumi \cite{ying2021xlumi} was proposed for blockchain systems. This protocol creates unidirectional channels. Unlike Lightning Network, xLumi drastically reduces the number of interactions and complexity of opening a payment channel. Users are not required to store a new secret for every off-chain transaction. However, xLumi has not been expanded to bidirectional channels and payment channel networks. It would be interesting to see how CryptoMaze can be adapted in xLumi based PCN.

\section{Proposed construction}
\label{5}
 $U_0$ wants to transfer $val$ coins to $U_r$ efficiently via the PCN $\mathcal{G}=(V,E)$, where $U_0 \in V, U_r \in V$. None of the nodes in the network must learn the identity of the payer, payee, or the coins transferred. Any honest party must not lose coins while routing the payment. We discuss the cryptographic preliminaries, system requirements, security, and privacy goals, followed by a formal description of the proposed protocol for realizing the payment.

\subsection{Cryptographic preliminaries}
(i) \emph{Discrete Logarithm Problem}: Given the elliptic curve $\mathbb{G}$ of order $q$ with base point $G$, $q=p^n$ where $p$ is a prime number and $n\in \mathbb{N}$, the discrete  logarithm  problem is  defined as follows: \emph{Given points $P,Q \in \mathbb{G}$, find an integer $a$ such that $Q=aP$, if $a$ exists. This computational problem is the \emph{Elliptic Curve Discrete Logarithm Problem} or ECDLP  that forms the fundamental building block for elliptic curve cryptography \cite{galbraith2016recent}}.\\
(ii) \emph{Standard Cryptographic Hash Function}: A cryptographic hash function is a one-way function that, given any fixed length input generates a unique fixed-length output. It is represented as $\mathcal{H}:\{0,1\}^*\rightarrow \{0,1\}^\lambda$, where $\lambda$ is the security parameter used in the model.

\subsection{System model}
\label{sys}
Given the PCN $\mathcal{G}=(V,E)$, a function $gain: V \rightarrow R^{+}$ is defined to quantify the coins any node has gained or lost while running an instance of the protocol. If we assume that the protocol starts at time $t_0$ and ends at time $t'$ then, for a node $v,\  gain(v)= \sum\limits_{u \in V, (v,u) \in E} C((v,u),t')-C((v,u),t_0)$. The global ideal functionality for blockchain $\mathcal{F}_{\mathbb{B}}$ \cite{malavolta} maintains $\mathbb{B}$. An arbitrary condition can be specified in the contract in order to execute a transaction in $\mathbb{B}$. $\mathcal{F}_{\mathbb{B}}$ is entrusted to enforce fulfillment of the contract before the corresponding transaction is executed. $t_{end}$ is the least timeout period set for an off-chain contract. $\Delta$ is the worst-case time taken for a transaction to settle on-chain. Each node $U_i \in V$ has its pair of the private key and public key. Pairs of honest users sharing a payment channel communicate using ideal functionality for secure message transmission $\mathcal{F}_{smt}$ \cite{can}. $U_i$ send $(sid,\textrm{instruction},U_i,U_j,m)$, containing the secret message $m$, to $U_j$ via $\mathcal{F}_{smt}$. $(sid,\textrm{instruction},U_i,U_j,|m|)$ is leaked to an adversary, where $|m|$ is the message length.

\textit{System Assumption.} Any user can get information on the network topology by sending a \emph{read} instruction to $\mathcal{F}_{\mathbb{B}}$. The latter sends the whole transcript of $\mathbb{B}$ in reply. The current value on each payment channel is not published but instead kept locally by the users sharing a payment channel. Every user is aware of the payment fees charged by other users in the PCN. All the nodes know each other's public keys. We do not discuss other problems occurring in the network like individual channel congestion, blocking of nodes, etc. These issues are orthogonal to the problem addressed in this paper. Problems arising due to concurrent payments can be addressed with the solutions proposed in \cite{malavolta}.

\textit{Communication Model.} We consider the bounded synchronous communication model \cite{attiya2004distributed}. In this model, time corresponds to the number of entries of $\mathbb{B}$, denoted by $|\mathbb{B}|$. Time is divided into fixed communication rounds. It is assumed that all messages sent by a user in a round are available to the intended recipient within a bounded number of steps in execution. The absence of a message in a round indicates an absence of communication from a user.


\subsection{Security and privacy goals}
\label{pg}
 We identify the following security and privacy notions:
\begin{itemize}[leftmargin=*]

\item \emph{Correctness}: Given all the nodes routing the payment are honest, $gain(U_0)= -(val+\sum\limits_{U_i \in V \setminus \{U_0,U_r\}: U_i \in \mathbb{PC}}f(U_i))$, $gain(U_r)=val$ and $gain(U_i)=f(U_i), \forall U_i \in V \setminus \{U_0,U_r\}$.
\item \emph{Consistency}: No intermediate node $U_i \in V \setminus \{U_0,U_r\}$ can provide the decommitment for the preceding off-chain contracts before the release of the decommitment in at least one of the succeeding off-chain contracts. If this holds, then no wormhole attack is possible as intermediate nodes cannot be bypassed.
\item \emph{Balance Security}: Honest intermediary does not lose coins, i.e., for any honest $U_i \in V\setminus \{U_0,U_r\}$, $gain(U_i)\geq 0$.
\item \emph{Value Privacy}: Corrupted users outside the
payment path must not have any information regarding the payment value in a pay operation involving only honest users.
\item \emph{Unlinkability}: Given a node $U_i$ splits the payments $val$ into $k$ parts $val_1,val_2,\ldots,val_k$ among the $k$ neighbors $U_{i,1},U_{i,2},\ldots,U_{i,k} : (U_i,U_{i,j}) \in E, j \in [1,k]$. If all the neighbors collude, they cannot figure out whether they are part of the same payment or a different payment.
\item \emph{Relationship Anonymity}: Given two simultaneous successful pay operations of the
form $(U_0,U_r,val)$ and $(U_0',U_r',val)$, using the same set of intermediate nodes and payment channels for routing payment, with at least one honest intermediate user $U_i$, corrupted intermediate users cannot determine whether the payment is from $U_0$ to $U_r$ or from $U_0'$ to $U_r'$ with a probability greater than $\frac{1}{2}$.
\item \emph{Atomicity}: If all the nodes preceding $U_r$ have forwarded their partial payments, then only the receiver can start claiming payments. Even if one of the nodes fails to forward the payment, then $gain(U_r)=0$ and $gain(U_i)=0, \forall U_i \in V\setminus \{U_0\}$.

\end{itemize}

\subsection{Mapping a set of paths into a set of edges}
\label{key}
In the example shown in \autoref{exanow}, $M$ wants to transfer an amount 5.1 units to $N$. Each intermediate node charges 0.1 unit as a processing fee. Initially, the set of routes must be realized by $M$. Any known routing algorithm like \cite{speedymurmur,yu2018coinexpress,lin2020rapido,wang2019flash} or \cite{mazumdar2020hushrelay} can be used. The paths returned are $p_1=\langle id_{M,A}\rightarrow id_{A,B}  \rightarrow id_{B,D} \rightarrow id_{D,N}\rangle$ and $p_2=\langle id_{M,A}  \rightarrow id_{A,C} \rightarrow id_{C,D} \rightarrow id_{D,N} \rangle$. Given that there are \emph{four} intermediate nodes, $M$ forwards 5.5 units to $A$, the latter will deduct 0.1 units, split the amount, and forwards 2.7 units each to channels $id_{A,B}$ and $id_{A,C}$. Node $B$ and $C$ charge 0.1 units each and forwards 2.6 units to channels $id_{B,D}$ and $id_{C,D}$ respectively. $D$ deducts 0.1 unit and forwards 5.1 units to $N$. In the paths $p_1$ and $p_2$, the channels $id_{M,A}$ and $id_{D,N}$ are shared. Instead of considering each path individually, a union of all the edges present in $p_1$ and $p_2$ is taken and set $\mathbb{PC}$ is constructed. The channels are inserted into the set in breadth-first order, starting from $M$. The set $\mathbb{PC}=\{ id_{M,A},id_{A,B},id_{B,D},id_{A,C},id_{C,D}, id_{D,N}\}$ is used as the protocol's input. Thus, mapping a set of paths into a set of edges allows a shared edge to appear not more than once in $\mathbb{PC}$.

\begin{figure}[!ht]
    \centering
    \includegraphics[scale=0.35]{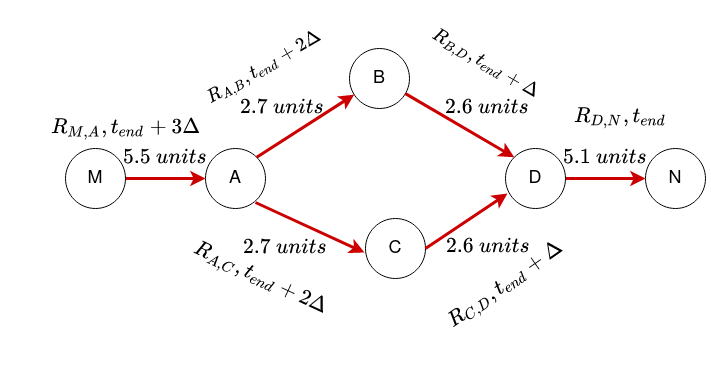}
    \caption{CryptoMaze executed on the network for routing payment from $M$ to $N$}
    \label{exanow}
\end{figure}

\subsection{Formal definition of the protocol}
\label{gen}
The protocol involves three phases: \emph{Preprocessing Phase}, \textit{Contract Forwarding Phase} and \textit{Release Phase}. $U_0$ forms the set $\mathbb{PC}$ and uses it as an input for \emph{Preprocessing Phase}. We define each phase in detail.

\subsubsection{Preprocessing phase}
\label{preprop}
$U_0$ extracts out the set of edges from $\mathbb{PC}$. We divide the phase into sub-phases, explained as follows:

(i) \emph{Secret value for claiming Payment.} The payee $U_r$ samples a random number $x_{\acute{r}}$ and sends $X_{\acute{r}}=x_{\acute{r}}G$ to $U_0$ via a secure communication channel. $U_0$ checks the number of incoming channels sending partial payments to $U_r$. If there are $k$ such channels, $U_0$ samples $y_i \in \mathbb{Z}_q$. The latter constructs the condition for each off-chain contract in reverse order, starting from node $U_r$. For any channel $id_{b,r} \in \mathbb{PC}, U_b \in V$, $R_{b,r}=e_{b,r} \sum\limits_{i=1}^k y_i G+X_{\acute{r}}$. $R_{b,r}$ is the condition encoded in the off-chain contract formed on the channel $id_{b,r}$. $e_{b,r}$ is blinding factor for hiding the secret value $y=\sum\limits_{i=1}^k y_i$. It is defined as $e_{b,r}=\mathcal{H}(\sum\limits_{i=1}^k y_i||id_{b,r})$. $U_r$ needs to provide the discrete logarithm of $R_{b,r}$ for claiming coins from $U_b$.

(ii) \emph{Conditions for off-chain contracts.} If any intermediate node $U_i$ is forwarding payment to a single node, $U_0$ samples independent strings $x_{i} \in \mathbb{Z}_q$ for the node. If a node $U_i$ forwards payments to multiple neighbors, then it must be ensured that $U_i$ does not lose coins when one of the neighbors fail to release the decommitment of an off-chain contract. To avoid this problem, our protocol uses a \emph{1-out-of-$m$} policy where even if one of the outgoing neighbors of $U_i$ responds, the latter can claim the coins. We first explain the procedure for computing secrets for a node that splits the payment value and forwards it to multiple neighbors with an example. 

In \autoref{exanow}, the condition used in the contracts established on each of the channels are denoted as follows: $R_{M,A}$ for $id_{M,A}$, $R_{A,B}$ for $id_{A,B}$, $R_{A,C}$ for $id_{A,C}$, $R_{B,D}$ for $id_{B,D}$, $R_{C,D}$ for $id_{C,D}$, and $R_{D,N}$ for $id_{D,N}$. Node $A$ splits the payment and sends it to nodes $B$ and $C$. The condition $R_{M,A}$ must be constructed so that the secrets provided by either $B$ or $C$ helps $A$ in claiming the amount from $M$. If $A$ establishes the same contract $R$ with nodes $B$ and $C$, then $R_{M,A}=R+e_{M,A} x_A G$. If $B$ and $C$ collude, they can link their payments. The situation is shown in \autoref{split}. To avoid the problem, two different conditions $R_{A,B}$ and $R_{A,C}$ are assigned to off-chain contracts on channels $id_{A,B}$ and $id_{A,C}$.  $A$ adjusts the value by adding $x_{A,B}G$ to $R_{A,B}$ and $x_{A,C} G$ to $R_{A,C}$ to ensure equality. Thus, we have $R_{M,A}=R_{A,B}+e_{M,A} x_A G+x_{A,B}G$ where $R_{A,B}+x_{A,B}G=R_{A,C}+x_{A,C}G$.

\begin{figure}[!ht]
    \centering
    \includegraphics[scale=0.33]{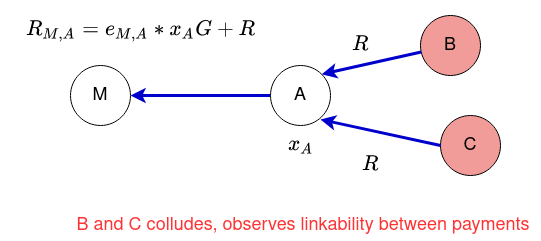}
    \caption{Problem of linkability between partial payments}
    \label{split}
\end{figure}


Let $Z=R_{A,B}+x_{A,B}G=R_{A,C}+x_{A,C}G$. If we fix the discrete logarithm of $Z$ to $x: Z=xG=R_{A,B}+x_{A,B}G=R_{A,C}+x_{A,C}G$, we can calculate the values $x_{A,B}$ and $x_{A,C}$. Again, $x_A=x_{A,B}+x_{A,C}$. Even if the off-chain contracts $R_{A,B},R_{A,C}$ and $R_{M,A}$ are settled on-chain, still a miner cannot establish linkability between the three. If $x_{A,B}$ and $x_{A,C}$ is not known to the miner, then it can establish a relationship between discrete logarithm of $R_{M,A}$ and discrete logarithm of $R_{A,B}$ or $R_{A,C}$ with negligible probability.

Summarizing the procedure, for a pair of channels $id_{i,j}$ and $id_{j,k}$, having conditions $R_{i,j}$ and $R_{j,k}$ where $U_j\neq U_r$, $e_{i,j}=\mathcal{H}(x_{j}||id_{i,j})$:

(a) If $U_j$ forwards payment to only one neighbor $U_k$, the condition $R_{i,j}$ is defined as follows:
\begin{equation}
\begin{matrix}
\label{qt1}
 R_{i,j}=e_{i,j} x_{j} G+R_{j,k} \\
\end{matrix}
\end{equation}  

(b) If $U_j$ splits the payment and forwards it to one of the neighbour $U_k$, $R_{i,j}$ is defined as follows:
\begin{equation}
\begin{matrix}
\label{qt2}
 R_{i,j}=e_{i,j} x_{j} G+R_{j,k}+x_{j,k}G \\
\end{matrix}
\end{equation}  
where $U_k \in \textrm{neighbor}(U_j), id_{j,k} \in \mathbb{PC}$.

To compute $x_{j,k}$ for $R_{j,k}$, $U_0$ generates a random value $\hat{x}$ such that $\hat{x}G+X_{\acute{r}}=R_{j,k}+x_{j,k}G, \forall U_k \in V, id_{j,k} \in \mathbb{PC}$. Fixing discrete logarithm as $\hat{x}$ helps $U_0$ to calculate $x_{j,k}$ for each channel $id_{j,k}$ corresponding to node $U_k$. The expression can be rewritten as follows:
\begin{equation}
\begin{matrix}
\label{q2}
 \qquad x_{j,k} G=X_{\acute{r}}+\hat{x} G-R_{j,k}
\end{matrix}
\end{equation}  

The discrete logarithm of $X_{\acute{r}}+\hat{x} G-R_{j,k}$ is known to $U_0$, i.e. $x_{j,k}=\hat{x}-\textrm{dlog}(R_{j,k}-{X_{\acute{r}}})$, where $\textrm{dlog}$ is the discrete logarithm. Once each $x_{j,k}$ gets computed, $U_0$ computes $x_j=\sum\limits_{  U_k \in V, id_{j,k} \in \mathbb{PC}} x_{j,k}$. Thus, for any node $U_i$ forwarding payments to multiple neighbors, the discrete logarithm for $R_{i,j}$ can be supplied by any of the outgoing neighbors of $U_j$. 
Substituting the value of $x_{j,k}G$ obtained from \eqref{q2} in \eqref{qt2}, we have:
\begin{equation}
 R_{i,j}=e_{i,j} x_{j} G+X_{\acute{r}}+\hat{x} G
\end{equation}


(iii) \textit{Setting timeout period.} The least timeout period assigned to the all incoming contract of $U_r$ is denoted as $t_{end}$. Starting from this point, the timeout period of all the preceding contracts get decided. For time-locked contracts established with any channel $id_{i,j}, U_j \neq U_r$, assign $t_{i,j}=\max\limits_{\forall U_k \in V, id_{j,k} \in \mathbb{PC}}\{t_{j,k} \}+\Delta$ as the timeout period of the contract on payment channel $id_{i,j}$. 

\subsubsection{Contract Forwarding Phase}
Each node $U_i$ uses shared variable $flag_i$ and $T_i$, both initialized to 0. The variable $flag_i$ is set to 1 if the node $U_i$ has received all the incoming contracts. $T_i$ is set to the current time when $U_i$ receives its first incoming contract request. $U_i$ waits for time $T_i+\delta$ to receive all the incoming contract requests, where $\delta>0$ is the latency. If the time elapsed is greater than $T_i+\delta$ but $flag_i$ is still 0, then $U_i$ sends abort to its preceding contracts, canceling the payment. 

Starting from node $U_0$, any node $U_i\neq U_r$ sends the request $(R_{i,j},val_{i,j},t_{i,j})$ for forming contracts to all its neighbor via $\mathcal{F}_{smt}$, once $flag_i$ is set to 1. For ease of analysis, we explain the procedure for one of its neighbors, say $U_j$. If the latter accepts the request, it gets the encrypted message $Z_{i,j}$. Upon decryption, it gets $M_j=\{(val_{j,k},x_{j,k},R_{j,k},t_{j,k},Z_{j,k}): \forall k \in V, id_{j,k} \in \mathbb{PC}\}$, where $Z_{j,k}$ is the encrypted message to be forwarded to the node $U_k$. $U_j$ checks the consistency of incoming contracts with the terms stated for an outgoing contract by calling the subroutine \textbf{TimeLockContractForward}, described in Module \ref{algo:t1}. The checks mentioned in this subroutine ensure the integrity of the phase. If the subroutine returns failure, then $U_j$ cancels all the off-chain contracts formed with preceding nodes. Else, $U_j$ waits for all preceding contracts such that the total value from the incoming contract is the summation of the fee charged by $U_j$ and the coins it needs to lock in all the outgoing contracts specified in $M_j$. After $U_j$ receives all the contracts within time $T_j+\delta$, then it begins forwarding the payment to its neighbor. The steps are defined in Procedure \ref{algo:payintmd}. The execution time is determined by the degree of the node and thus the time complexity of the procedure is $O(\frac{|E|}{|V|})$ where $|E|$ is the number of edges and $|V|$ is the number of vertices in $\mathcal{G}$. 

A node can identify its predecessor if the former obtains similar messages upon decryption. The node can forward one such message and discard the rest. This phase continues till all partial payments reach $U_r$. Even if there is one off-chain contract that did not get instantiated in a payment channel belonging to $\mathbb{PC}$, $U_r$ cannot compute the secret $y$ for claiming coins. Satisfying this constraint implies that all the partial payments have been combined properly, guaranteeing \emph{atomicity}. Once the receiver has received all the partial payments within a bounded amount of time, it triggers the \emph{Release Phase}.

\subsubsection{Release Phase}
$U_r$ gets the secret share from all the incoming off-chain contracts forwarding the payment. The former can compute the secret $y$ as described in Procedure \ref{algo:releasereceiver}. Upon computing the secret, $U_r$ calls the subroutine \textbf{TimeLockContractRelease} defined in Module \ref{algo:t2}. The module returns the solution for the condition encoded in the incoming contracts forwarded by its neighbor. If the solution is correct, $U_r$ sends a decision of acceptance to its predecessor along with the secret. Else, it sends an abort message to the neighbors and the payment fails. The abort process is mentioned in Procedure \ref{algo:abort}. Any intermediate node, involved in forwarding conditional payment can claim the coins if at least one of the neighbors responds. The steps followed by an intermediate node for claiming payment have been defined formally in Procedure \ref{algo:releaseintmd}. Time complexities of Procedure \ref{algo:releasereceiver} and Procedure \ref{algo:releaseintmd} are $O(\frac{|E|}{|V|})$ each.

\begin{module}[!ht]
    \SetKwInOut{Input}{Input}
    \SetKwInOut{Output}{Output}

    \caption{TimeLockContractForward for node $U_j \in V$}
        \label{algo:t1}

\textbf{Input } : $(D_j,t_{i,j},R_{i,j})$ \\
Parse $D_j=\{(id_{j,k},x_{j,k},R_{j,k},t_{j,k}): \forall k \in V, id_{j,k} \in E\}$\\
Compute $x_j=\sum\limits_{ k \in V, id_{j,k} \in E} x_{j,k}$\\
Compute $e_{i,j}=\mathcal{H}(x_{j}||id_{i,j})$\\
\If{$|D_j|>1$}
{
\For{$k \in V: id_{j,k} \in E$}
{

\If{$R_{i,j}\stackrel{?}{=}e_{i,j} x_{j}G+R_{j,k}+x_{j,k}G$ and $t_{i,j}\stackrel{?}{\geq} t_{j,k}+\Delta$}
{
    \emph{continue}\\
}
\Else
{
\Return failure
}
}
}
\Else
{
\If{$R_{i,j}\neq e_{i,j} x_{j}G+R_{j,k}$ or
 $t_{i,j}< t_{j,k}+\Delta$}
 {
 \Return failure
 }

}
   \Return success

\end{module}

\begin{module}[!ht]
    \SetKwInOut{Input}{Input}
    \SetKwInOut{Output}{Output}

    \caption{TimeLockContractRelease for node $U_j \in V$}
        \label{algo:t2}

\textbf{Input } : $(r_{j,k},x_j,id_{i,j})$ \\

Compute $e_{i,j}=\mathcal{H}(x_{j}||id_{i,j})$\\
Compute $r_{i,j}=e_{i,j} x_{j}+r_{j,k}$\\
\Return $r_{i,j}$\\
            
\end{module}

\section{An example of CryptoMaze}
We present a complete flow of our protocol with an example. $M$ wants to transfer an amount of 5.1 units to $N$, shown in Fig. \ref{exanow}. Each intermediate node charges 0.1 units as a processing fee. It is assumed that each node has a public key and a private key generated at the time of joining the network. First, we mention how the set of paths returned by any standard routing algorithm must be mapped into a set of edges. The set $\mathbb{PC}$ is constructed using the information. Next, we provide a detailed construction of CryptoMaze where $\mathbb{PC}$ is used as an input.

\subsection{Mapping a set of paths into a set of edges}
The paths returned are $p_1=\langle id_{M,A}\rightarrow id_{A,B}  \rightarrow id_{B,D} \rightarrow id_{D,N}\rangle$ and $p_2=\langle id_{M,A}  \rightarrow id_{A,C} \rightarrow id_{C,D} \rightarrow id_{D,N} \rangle$. Given that there are 4 intermediate nodes, $M$ forwards 5.5 units to $A$, the latter will deduct 0.1 units, split the amount and forwards 2.7 units each to channels $id_{A,B}$ and $id_{A,C}$. They will charge 0.1 unit and forward it to channels $id_{B,D}$ and $id_{C,D}$. $D$ deducts 0.1 unit and forwards the rest to $N$. We see that in paths $p_1$ and $p_2$, the channels $id_{M,A}$ and $id_{D,N}$ are shared. Instead of considering each path individually, a union of all the edges present in $p_1$ and $p_2$ is taken and set $\mathbb{PC}$ is constructed. The channels are inserted into the set in a breadth-first order, starting from $M$. The set $\mathbb{PC}=\{ id_{M,A},id_{A,B},id_{B,D},id_{A,C},id_{C,D}, id_{D,N}\}$ serves as an input for the protocol execution.

\begin{figure}[!ht]
    \centering
    \includegraphics[scale=0.35]{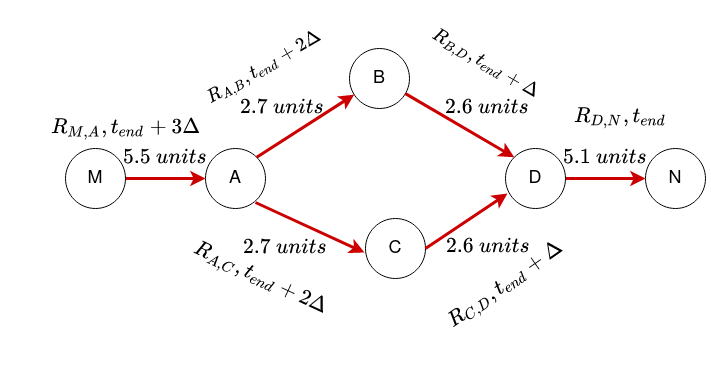}
    \caption{An instance of the protocol is executed}
    \label{exanow}
\end{figure}

\subsection{Phases of the protocol}
$M$ extracts the edges from the set $\mathbb{PC}$. We define each phase of the algorithm in layman terms: 
\begin{itemize}[leftmargin=*]

\item \textbf{Preprocessing Phase}: $M$ samples random values $x_{B},x_{C},x_{D}$ for nodes $B,C$ and $D$. Since, $A$ splits the payment and forwards it to two channels $(A,B)$ and $(A,C)$.  $x_{A,B}$ is assigned to channel $id_{A,B}$ and $x_{A,C}$ is assigned to channel $id_{A,C}$. From these values, we get $x_A=x_{A,B}+x_{A,C}$. $N$ samples a secret $x_N$ but this is not shared with $M$.\\

(i) \textit{Secret Value for claiming payment.} If receiver $N$ receives $k$ such partial payments, then $M$ samples $y_1,y_2,\ldots,y_k$. The secret value $y=\sum\limits_{i=1}^k y_i$ will be used for claiming payments, which will be discussed later. The motivation behind this operation is to prevent the receiver from claiming payments until and unless it has accepted the off-chain contracts corresponding to all partial payments. Since $N$ has one in-degree, a single secret $y$ is sampled by $M$.\\

(ii) \textit{Condition for off-chain contract.} For a given channel $id_{i,j}$,a blinding factor $e_{i,j}$ is constructed for hiding the secret $x_i$. The condition used in the off-chain contract of channel $id_{M,A}$ is $R_{M,A}$, for channel $id_{A,B}$ is $R_{A,B}$, for channel $id_{A,C}$ is $R_{A,C}$, for channel $id_{B,D}$ is $R_{B,D}$, for channel $id_{C,D}$ is $R_{C,D}$, and for channel $id_{D,N}$ is $R_{D,N}$. This is represented in Fig. \ref{exanow}. The computation of conditions needed for each off-chain contracts formed in the channels in $\mathbb{PC}$, is discussed below: 
\begin{equation}
\label{contract}
\begin{matrix}
R_{D,N}=e_{D,N}yG+x_{N}G, \ e_{D,N}=\mathcal{H}(y||id_{D,N})\\
R_{B,D}=R_{D,N}+e_{B,D} x_D G,\   e_{B,D}=\mathcal{H}(x_{D}||id_{B,D})\\
R_{C,D}=R_{D,N}+e_{C,D} x_D G, \   e_{C,D}=\mathcal{H}(x_D||id_{C,D})\\
R_{A,B}=R_{B,D}+e_{A,B} x_B G,  \   e_{A,B}=\mathcal{H}(x_B|id_{A,B})\\
R_{A,C}=R_{C,D}+e_{A,C} x_C G, \   e_{A,C}=\mathcal{H}(x_C||id_{A,C})\\
\end{matrix}
\end{equation}
Till this point, all the nodes were forwarding payment to a single neighbor. However, a node might not be able to route the entire payment value through one single channel. In that case, it is better to split the payment across several outgoing payment channels. Node $A$ has to split the payment across channels $(A,B)$ and $(A,C)$. It is quite possible that one of the neighbors fails to resolve the contract and doesn't release the secret. In that case, $A$ is at a loss if the protocol requires both $B$ and $C$ to respond to resolve the condition $R_{M,A}$. \emph{Balance security} gets violated. To avoid this problem, our protocol uses a \emph{1-out-of-m} policy where even if one of the outgoing neighbors of $A$ respond, the latter can claim payment. The income of $A$ is either equal to or greater than the expenditure. We briefly discuss the underlying concept of computing the contract $R_{M,A}$ so that the secrets provided by either node $B$ or $C$ helps $A$ in claiming money from $M$. Let us discuss a naive approach. If $A$ forms the same contract $R$ with nodes $B$ and $C$, then:

\begin{equation}
\label{init}
\begin{matrix}
R_{M,A}=R+e_{M,A} x_A G\\
\end{matrix}
\end{equation}
\begin{figure}[!ht]
    \centering
    \includegraphics[scale=0.35]{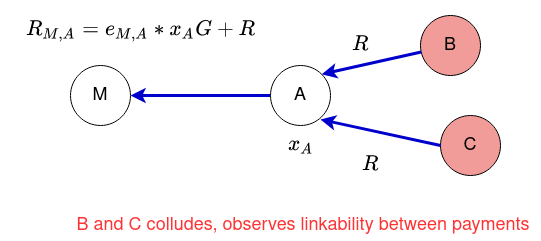}
    \caption{Contract forwarded by the node where the split occurs}
    \label{split}
\end{figure}

If $B$ and $C$ collude, then can figure out that they are part of the same payment, violating the property of \emph{unlinkability}. The problem is shown in Fig. \ref{split}. Hence, the contracts assigned to channels $id_{A,B}$ and $id_{A,C}$ must be different. Thus, we have different conditions $R_{A,B}$ for channel $id_{A,B}$ and $R_{A,C}$ for channel $id_{A,C}$ respectively. $A$ adds $x_{A,B} G$ to $R_{A,B}$ and $x_{A,C}G$ to $R_{A,C}$ so that the following condition holds:
\begin{equation}
\label{d1}
R_{A,B}+x_{A,B}G=R_{A,C}+x_{A,C}G\\
\end{equation}
From \autoref{d1}, we can write $R_{M,A}=R_{A,B}+e_{M,A} x_A G + x_{A,B}G\\=R_{A,C}+e_{M,A} x_{A} G+ x_{A,C} G$, where $e_{M,A}=\mathcal{H}(x_A||id_{M,A})$.

If $M$ fixes the discrete logarithm of $R_{A,B}+x_{A,B}G-x_N G$ to an $x$:
\begin{equation}
 xG=R_{A,B}-x_N G+x_{A,B} G=R_{A,C}-x_N G+x_{A,C} G
\end{equation}
The expression can be rewritten as:
 \begin{equation}
 \begin{matrix}
x_{A,B}G=xG+x_N G-R_{A,B}\\
x_{A,C}G=xG+x_N G-R_{A,C}\\
\end{matrix}
\end{equation}

From \autoref{contract}, we have,
\begin{equation}
\label{wq2}
\begin{matrix}
R_{A,B}=R_{B,D}+e_{A,B} x_B G\\
=R_{D,N}+e_{B,D} x_D G+e_{A,B} x_B G\\
=e_{D,N} y G+ x_{N}G+e_{B,D}x_D G+e_{A,B} x_B G\\
\end{matrix}
\end{equation}
\begin{equation}
\label{wq1}
\begin{matrix}
R_{A,C}=R_{C,D}+e_{A,C} x_C G\\
=R_{D,N}+e_{C,D} x_D G+e_{A,C} x_C G\\
=e_{D,N} y G+ x_{N} G+e_{C,D} x_DG + e_{A,C} x_C G\\
\end{matrix}
\end{equation}
From \autoref{wq2} and \autoref{wq1}, we get $R_{A,B}-x_N G=e_{D,N} y G+ e_{B,D}x_D G+e_{A,B} x_B G$ and $R_{A,C}-x_N G=e_{D,N} y G+ e_{C,D} x_DG + e_{A,C} x_C G$. 
Thus, we have 
\begin{equation}
\label{9}
\begin{matrix}
dlog(R_{A,B}-x_N G)= e_{D,N} y + e_{B,D}x_D +e_{A,B} x_B \\
dlog(R_{A,C}-x_N G)= e_{D,N} y + e_{C,D}x_D +e_{A,C} x_C \\

\end{matrix}
\end{equation}

where dlog is the discrete logarithm. Since $M$ knows $x_1=dlog (R_{A,C}-x_N G)$ and $x_2= dlog(R_{A,B}-x_N G)$, the value of $x_{A,B}$ and $x_{A,C}$ can be calculated.
 \begin{equation}
 \begin{matrix}
  x_{A,B}=x-x_1\\
  x_{A,C}=x-x_2
   \end{matrix}
 \end{equation}

$M$ can compute $x_{A}=x_{A,B}+x_{A,C}$. The condition $R_{M,A}$ for channel $id_{M,A}$ can be computed in the following way:
\begin{equation}
\begin{matrix}
R_{M,A}=R_{A,B}+x_{A,B}G+e_{M,A} x_AG\\
=R_{A,C}+x_{A,C}G+e_{M,A}x_AG\\
\end{matrix}
\end{equation}
where $e_{M,A}=\mathcal{H}(x_A||id_{M,A})$.

(iii) \textit{Setting Timeout Period.} The timeout period of each contract doesn't follow a linear relation, since the amount split and merges at certain points. Since the worst-case time taken for a transaction to settle on-chain is $\Delta$, the difference in timeout period between incoming and outgoing off-chain contracts must be at least $\Delta$. Since the contracts accepted by the receiver must have the least timeout, the assignment is done in reverse order. Hence, $R_{D,N}$ is assigned a timeout of $t_{end}$. Next timeout period for $R_{B,D}$ and $R_{C,D}$ is $t_{end}+\Delta$, for $R_{A,B}$ and $R_{A,C}$ is $t_{end}+2\Delta$. The timeout period for $R_{M,A}=\Delta+max(\textrm{timeout}(R_{A,B})+\textrm{timeout}(R_{A,C}))=t_{end}+3\Delta$.

\item \textbf{Contract Forwarding Phase}:
$M$ uses onion routing for forwarding the off-chain contract, with each message encrypted by the public key of the intermediate nodes. 
We describe each step as follows:
\begin{itemize}[leftmargin=*]
\item $M$ sends a request to form contract $R_{M,A}$ with timeout period $t_{M,A}$ to $A$. The amount  forwarded is $val_{M,A}$=5.5 units. If $A$ accepts the contract, then it forwards the encrypted date to $A$.
\item $A$ decrypts the message and finds secret values $x_{A,B}$ and $x_{A,C}$ for the channels $id_{A,B}$ and $id_{A,C}$. It computes $x_A=x_{A,B}+x_{A,C}$. Along with this, it finds instruction to forward the contracts $R_{A,B}$ and $R_{A,C}$ to $B$ and $C$. The amount to be forwarded $val_{A,B}$, $val_{A,C}$, and timelock of each contract $t_{A,B},t_{A,C}$, is mentioned as well. $A$ computes the blinding factor $e_{M,A}=\mathcal{H}(x_A||id_{M,A})$, checks the conditions $R_{M,A}\stackrel{?}=R_{A,B}+x_{A,B}G+e_{M,A}x_A G$ and $R_{M,A}\stackrel{?}=R_{A,C}+x_{A,C}G+e_{M,A}x_A G$, as stated in Eq.\ref{contract}. It also checks whether $t_{M,A}\stackrel{?}{\geq} \max{(t_{A,B},t_{A,C})}+\Delta$ and $val_{M,A}=f(A)+val_{A,B}+val_{A,C}$, $f(A)$=0.1 units is processing the fee charged by $A$, $val_{A,B}$=2.7 units is the conditional payment forwarded to $B$ and $val_{A,B}$=2.7 units is the conditional payment forwarded to $C$.
\item If the checks mentioned in the previous step holds, $A$ sends contract formation request to $B$ and $C$. If both of them agree to form the contract, then contracts $R_{A,B}$ is formed on channel $id_{A,B}$ and $R_{A,C}$ is formed on channel $id_{A,C}$. $A$ forwards the encrypted data to each of the nodes. 
\item $B$ decrypts and gets $R_{B,D}$, secret value $x_{B}, t_{B,D}, val_{B,D}$ and $C$ gets $R_{C,D}$, secret value $x_{C}, t_{C,D}, val_{C,D}$. Each of them computes blinding factors $e_{A,B}$ and $e_{A,C}$. $B$ checks the condition $R_{A,B}\stackrel{?}{=}R_{B,D}+e_{A,B}x_B G$, $t_{A,B}\stackrel{?}{\geq} t_{B,D}+\Delta$ and $val_{A,B}=f(B)+val_{B,D}$, $val_{B,D}$=2.6 units is the conditional payment forwarded to $D$. $C$ checks the condition $R_{A,C}\stackrel{?}{=}R_{C,D}+e_{A,C}x_C G$, $t_{A,C}\stackrel{?}{\geq} t_{C,D}+\Delta$ and $val_{A,C}=f(C)+val_{C,D}$, $val_{C,D}$=2.6 units is the conditional payment forwarded to $D$.

\item Both $B$ and $C$ find that the next destination is $D$. They send request to form the contract $R_{B,D}$ and $R_{C,D}$. If $D$ accepts the contract, then it receives encrypted messages from both parties.
\item $D$ decrypts both the messages and finds message $R_{D,N}$, secret value $x_{D}, t_{D,N}$ and $val_{D,N}$. It computes blinding factor $e_{B,D}$ and $e_{C,D}$, checks the condition for $R_{B,D}$ and $R_{C,D}$ as stated in Eq. \ref{contract}. Next, $D$ checks $val_{C,D}+val_{B,D}=f(D)+val_{D,N}$, $val_{D,N}$=5.1 units, and the consistency of timeout period $t_{B,D}\stackrel{?}{\geq} t_{D,N}+\Delta,t_{C,D}\stackrel{?}{\geq} t_{D,N}+\Delta$.
\item Since $D$ gets the same message from $B$ and $C$, it discards one and forwards the contract formation request to $N$. Once $N$ accepts the off-chain contract $R_{D,N}$ with timeout period $t_{D,N}$, $D$ forwards the encrypted packet. The receiver decrypts the packet to find the message $y$ and $t_{end}$. If $N$ had $k$ such incoming off-chain contracts, then each would have forwarded the value $y_1,y_2,\ldots,y_k$. In that case, $N$ adds all these partial secrets to get the value $y$.
\end{itemize}

\item \textbf{Release Phase}: The receiver, upon accepting the off-chain contracts for all partial payments, gets the secret value $y$ forwarded by sender. It checks $t_{D,N}=t_{end}$ and $val_{D,N}$=5.1 units. Next, the payer computes $e_{D,N}=\mathcal{H}(y||id_{D,N})$, $r_{D,N}=e_{D,N}y+x_N$ and sends it to $D$. The latter uses $r_{D,N}$ to compute $r_{B,D}=e_{B,D}x_{D}+x_{D}$ and $r_{C,D}=e_{C,D}x_{D}+x_{D}$ to claim 2.6 units each from $B$ and $C$, where $e_{B,D}=\mathcal{H}(x_D||id_{B,D})$ and $e_{C,D}=\mathcal{H}(x_D||id_{C,D})$. $B$ uses $r_{B,D}$ to compute $r_{A,B}=e_{A,B}x_{B}+r_{B,D}$, where $e_{A,B}=\mathcal{H}(x_B||id_{B,D})$. It claims 2.7 units from $A$ by releasing $r_{A,B}$. $C$ uses $r_{C,D}$ to compute $r{A,C}=e_{A,C}x_{C}+r_{C,D}$, where $e_{A,C}=\mathcal{H}(x_C||id_{C,D})$. However, if $C$ decides not to respond, then $A$ can still claim the payment by using the secret released by $B$. It computes $r_{M,A}=r_{A,B}+x_{A,B}+e_{M,A}x_A$ and claims 5.5 units from $M$.


\end{itemize}

\section{Security definition of CryptoMaze}
\label{attacksec}

For modeling security and privacy definition of payment across several payment channels under concurrent execution of an instance of \textit{CryptoMaze}, we take the help of Universal Composability framework, first proposed by Canetti et al. \cite{can}. Notations used here are similar to \cite{malavolta}. 

\subsection{Attacker model \& assumptions}
The real-world execution of the protocol is attacked by an adversary $\mathcal{A}$, a \texttt{PPT}, or \emph{probabilistic polynomial-time} algorithm. We assume that only static corruption is allowed, i.e., the adversary must specify the nodes it wants to corrupt before the start of the protocol \cite{fairswap}, \cite{malavoltamulti}. Once a node is corrupted, $\mathcal{A}$ gets access to its internal state and controls any transmission of information to and from the corrupted node.  The attacker is provided with the internal state of the corrupted node. Also, the incoming and outgoing communication of such a node gets routed through $\mathcal{A}$.

\subsection{Ideal world functionality}
\label{basicop}

\begin{figure}[!ht]
    \centering
    \includegraphics[scale=0.26]{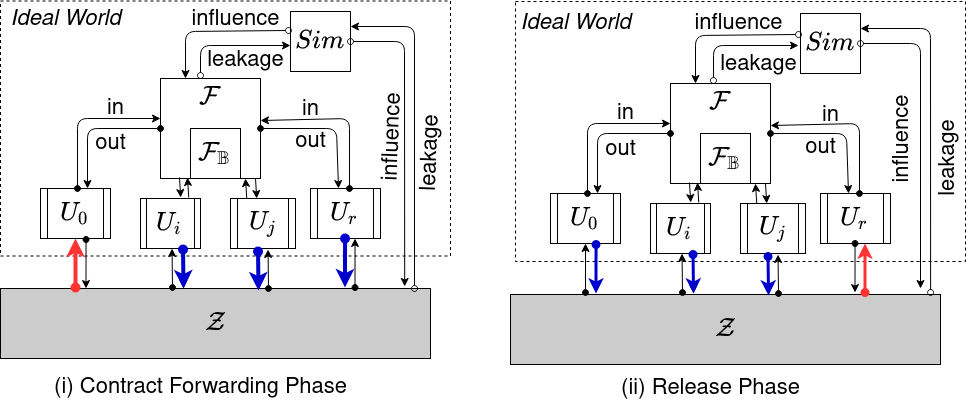}
    \caption{Execution of $\mathcal{F}$ with dummy parties $U_0$, $U_r$ representing payer and payee, $U_i$,$U_j$ representing intermediaries routing payment }
    \label{state}
\end{figure}

\textit{Notations.} We define an ideal functionality $\mathcal{F}$ for payment in PCN. Honest nodes in the network are modeled as interactive Turing machines. Such nodes are termed as dummy parties and they can communicate with each other via $\mathcal{F}$. $U_0$ denotes the initiator of the protocol and $U_r$ denotes the receiver. The latter internally access the global ideal functionality $\mathcal{F}_{\mathbb{B}}$, defined in Section \ref{sys}. Any payment channel existing in $\mathbb{B}$ is denoted by $(id_{i,j},v_{i,j},t_{i,j}',f_{i,j})$, where $id_{i,j}$ is the channel identifier of the payment channel existing between dummy parties $U_i$ and $U_j$, $v_{i,j}$ is the capacity of the channel, $t_{i,j}'$ is the expiration time of the channel and $f_{i,j}$ is the associated fee charged for the channel $id_{i,j}$. $\mathcal{F}$ maintains two lists internally - one for keeping track of the list of closed channels, denoted by $\mathcal{C}$, and one for keeping track of the list of off-chain payments, denoted by $\mathcal{L}$ \cite{malavolta}. Upon executing an off-chain payment in the channel $id_{i,j}$, $(id_{i,j},v_{i,j}',t_{i,j},h_{i,j}')$ is entered into $\mathcal{L}$ where $v_{i,j}'$ is the residual capacity of the channel and $t_{i,j}$ is the expiration time of the payment, $h_{i,j}'$ is the event identifier. When a channel $id_{i,j}$ is closed on-chain, it is entered into the list $\mathcal{C}$. Payment channels  forwarding the payment from $U_0$ to $U_r$ are put in set $\mathbb{PC}$, added serially upon breadth-first traversal of the network, starting from $U_0$. The flow in each channel $id_{i,j}$ present in $\mathbb{PC}$ is denoted by $val_{i,j}$.

\begin{proc}[!ht]
    \SetKwInOut{Input}{Input}
    \SetKwInOut{Output}{Output}

    \caption{Contract Forwarding Phase for node $U_j \in V$}
        \label{algo:payintmd}

   Upon input  $(forward,m)$ from $U_i$, parse $m$ to get $R_{i,j},val_{i,j},t_{i,j}$.\\
   Initialize \emph{proceed}=0\\
                $val_{j}=val_{j}+val_{i,j}$\\
        Form contract with $U_i$ using condition $R_{i,j}$, receive $Z_{i,j}$ from $U_i$.\\
			
         \If{$U_j \neq U_r$}
         {
         Decrypt $Z_{i,j}$ to get $M_{j}=\{(id_{j,k},val_{j,k},x_{j,k},R_{j,k},t_{j,k},Z_{j,k}): \forall k \in V, id_{j,k} \in E\}$\\
         
           Form set $D_j=\{(id_{j,k},x_{j,k},R_{j,k},t_{j,k}): \forall k \in V, id_{j,k} \in E\}$\\
          Call \textbf{TimeLockContractForward} with $(D_j,t_{i,j},R_{i,j})$ as the input\\
                    \If{(receives success)}
                    {
                    	Set \emph{proceed}=1\\
                    	\If{$T_{j}=0$}
                    	{
                    	  Set $T_j=T_{current}$\\
                    	}
                    }
                      \If{\emph{proceed}=1}
          {
          		
          
            \If{$val_{j}< \sum\limits_{ U_k \in V, id_{j,k} \in E} val_{j,k} + f(U_j)$} 
            {
            		Wait for timeperiod of $T_j+\delta$\\
				    \If{timeperiod has elapsed and $flag_{j}=0$}
				    {
				    		Set \emph{proceed}=0\\
				    } 
%
%
            		
            }
            \Else
            {
            	Set $flag_{j}=1$\\
                \For{$U_k \in V: id_{j,k} \in D_{j}$}
                {
                		 Send $(\textrm{forward},R_{j,k},val_{j,k},t_{j,k})$ to $U_k$, receive response from $U_k$\\
                		 {

					  		$C(U_j,U_k)=C(U_j,U_k)-val_{j,k}$\\
					  		                		    Set Contract$(id_{j,k})$=1, send $Z_{j,k}$ to $U_{k}$\\
					  		
                		 }
%
%
%
%
%
                }
            }
             
           }
        
           \If{\emph{proceed}=0}
           {
             \For{$U_m \in V: id_{m,j} \in E$}
                {
			
	\If{ $isContract(id_{m,j})=1$}
								{
										Send $(abort)$ to $U_m$.
					
			}
			}

           }

                    }

          
           \end{proc}

    \begin{proc}[!ht]
\setcounter{AlgoLine}{38}
                    \Else
                    {
                                        	\If{$T_{r}=0$}
                    	{
                    	  Set $T_r=T_{current}$\\
                    	}

                    \If{$val_r<val$}
                    {
                    
						Wait for timeperiod of $T_r+\delta$.\\
										    \If{timeperiod has elapsed and $flag_{r}=0$}
				    {
				    		\For{$U_b \in V: id_{b,r} \in E$}
{
\If{$isContract(id_{b,r})=1$}
{
   Send $(abort)$ to $U_b$
   }
}

				    }

                    }
                    \Else
                    {
                    Set $flag_r=1$.\\
                                            Call Release Phase defined in Procedure \ref{algo:releasereceiver}
                    }
                    
                    }
\end{proc}

\subsubsection{Operations}
We describe the operation \texttt{PAY} in the ideal world. $\mathcal{F}$ initializes a pair of local empty lists $(\mathcal{L},\mathcal{C})$. Each session is denoted by a session identifier $sid$. The phase is initiated by $U_0$, sending the payment value $val$ to be paid to $U_r$, the least timeout period of off-chain contract $t_{end}$. The other inputs are the set of payment channels $\mathbb{PC}$ along with the flow in each channel $val_{i,j}$ and the timeout period of off-chain contracts established on each channel, denoted as $t_{i,j}$, $\forall id_{i,j} \in \mathbb{PC}$. $\mathcal{F}$ initializes the variable $contract(sid,id_{i,j})=0$, $\forall id_{i,j} \in \mathbb{PC}$ to indicate that till now no off-chain contract got established in this session. \texttt{PAY} is divided into two phases: (i) \emph{Contract Forwarding Phase} and (ii) \emph{Release  Phase}, defined in \autoref{fpcn}. 

(i) The \emph{Contract Forwarding Phase} is triggered after $U_0$ sends the \emph{pay} instruction along with the set $\mathbb{PC}$ and the value of the payment. The inputs provided from the environment $\mathcal{Z}$ are marked in red, as shown in \autoref{state}(i). Each node $U_i \neq U_r, U_i \in \mathbb{PC}$ is visited in breadth-first fashion and the nodes are inserted in the queue $Q_{pay}$. Before sending the request to $U_j$, $\mathcal{F}$ checks whether an open channel $id_{i,j}$ exists in $\mathbb{B}$. Next, it checks whether the channel $id_{i,j}$ has enough capacity for forwarding the payment. The consistency of the timeout period for incoming and outgoing contracts is checked as well. If any of the conditions fail, $\mathcal{F}$ removes any entry for off-chain payments in $\mathcal{L}$ and aborts. If all the criteria hold, $\mathcal{F}$ forwards the partial payment to node $U_j$, output arrow marked in blue, shown in \autoref{state}(i). If all preceding contracts of $U_j$ got established, then it becomes a candidate for forwarding the payment. $U_j$ is thus inserted into $Q_{pay}$. If $U_j$ sends abort, then all the entries in $\mathcal{L}$ are removed and $\mathcal{F}$ aborts.

(ii) Once the payment reaches $U_r$, it triggers the \emph{Release Phase} by sending a \emph{response} to $\mathcal{F}$, input arrows marked in red, shown in \autoref{state}(ii). If $U_r$ sends abort, then the payment is considered to have failed. All the entries are removed from $\mathcal{L}$ and $\mathcal{F}$ aborts. If $U_r$ responds with success, then $\mathcal{F}$ sends a success message to predecessors of $U_r$, updates the entry in $\mathcal{L}$. The output of the intermediate parties sent to environment $\mathcal{Z}$ is marked as blue arrows in \autoref{state}(ii). If the predecessor sends an abort message, then such a node is marked as \emph{visited} and the entry is pushed in $Q_{failure}$. Else, that node is considered as the candidate for forwarding the success message to its predecessors and marked as \emph{visited}, if it has not been visited before. Nodes in $Q_{failure}$ are dealt with later after all the successful payments get settled. Each of these nodes sends an abort message to its predecessor. If a predecessor has not visited before, then it is pushed in $Q_{failure}$ and the process continues.
\begin{proc}[!ht]
    \SetKwInOut{Input}{Input}
    \SetKwInOut{Output}{Output}

    \caption{Release Phase for receiver $U_r$}
        \label{algo:releasereceiver}
            Set $stop=0$\\
            
\For{$U_b \in V: id_{b,r} \in E$}
{
            
              Decrypt $Z_{b,r}$ to get $\{y_{b,r},t_{end}\}$\\
        
        \If{$t_{b,r}=t_{end}$}
        {

				$y=y+y_{b,r}$\\

		}
		\Else
		{
            $stop$=1\\
              \textbf{break from the loop}\\

		}

}
\If{stop=0}
{
\For{$U_b \in V: id_{b,r} \in E$}
{
Call \textbf{TimeLockContractRelease} with $(y,x_{\acute{r}},id_{b,r})$ as input, gets $r_{b,r}$.\\
           
           \If{$R_{b,r} \neq r_{b,r}G$}   
           {
                         $stop$=1\\
              \textbf{break from the loop}\\

           }
           \Else
           {
           		Store $r_{b,r}$.
           }
			
}
}
\If{$stop=1$}
{
\For{$U_b \in V: id_{b,r} \in E$}
{
\If{$isContract(id_{b,r})=1$}
{
   Send $(abort)$ to $U_b$
   }
}
}
\Else
{
\For{$U_b \in V: id_{b,r} \in E$}
{
   Send $(accept,r_{b,r})$ to $U_b$
}

}
	           
\end{proc}

\subsubsection{Discussion}
 The operation \texttt{PAY} defined in ideal functionality $\mathcal{F}$ satisfies privacy properties of CryptoMaze in the following ways:

\begin{itemize}[leftmargin=*]

\item Correctness: In \emph{Contract Forwarding phase}, each intermediate node $U_i$ gets instructions for forwarding payment from $\mathcal{F}$ on behalf of node $U_j$, provided $\sum\limits_{U_k \in V, id_{k,j} \in \mathbb{PC}} val_{k,j}=\sum\limits_{U_m \in V, id_{j,m} \in \mathbb{PC}} val_{j,m}+f(U_m)$. $U_r$ triggers the release phase and responds with success, provided it has received the amount $val$. If all the parties have behaved honestly and $U_r$ responds with success in the \emph{release phase}, then $\mathcal{F}$ updates in $\mathcal{L}$ the channels present in $\mathbb{PC}$. Thus, $U_0$ can complete the payment by forwarding $val+\sum\limits_{U_i \in V\setminus \{U_0,U_r\}}f(U_i)$, where each node $U_i \in V \setminus \{U_0,U_r\}$ gains $f(U_i)$ and $U_r$ gets the amount $val$.
\item Consistency:  Release Phase defined in \autoref{fpcn} shows that $U_i$ is pushed into the queue $\mathcal{T}$ only if there is a successor $U_j$ that had resolved the off-chain contract forwarded by $U_i$. Once $U_i$ enters into $\mathcal{T}$, then it will be popped out of the queue for resolving its preceding contracts. If all the neighbors of $U_i$ have sent abort, then none of the preceding contracts forwarded to $U_i$ will get resolved.
\begin{proc}[!ht]
    \SetKwInOut{Input}{Input}
    \SetKwInOut{Output}{Output}

    \caption{Release Phase for node $U_j \in V \setminus \{U_r\}$}
        \label{algo:releaseintmd}

			   Upon receiving input $(U_i,accept,m)$, parse $m$ to get $r_{j,i}$\\

			    $C(U_i,U_j)=C(U_i,U_j)+val_{j,i}$\\
			\If{$release_j$=0}
			{
			   Set $release_j$=1\\
			   \If{$U_j$ had forwarded payment to more than one node}
			   {
			   		$r_{j,i}=r_{j,i}+x_{j,i}$\\
			   		$x_j=\sum\limits_{ U_i \in V: id_{j,i} \in E}x_{j,i}$
			   		
			   }
			                \For{$U_m \in V: id_{m,j} \in E$}
                {

			Call \textbf{TimeLockContractRelease} with $(r_{j,i},x_{j},id_{m,j})$ as input, gets $r_{m,j}$.\\
			Send $(\textrm{accept},r_{m,j})$ to $U_m$.
			}
			   
			}
\end{proc}

\begin{proc}[!ht]
    \SetKwInOut{Input}{Input}
    \SetKwInOut{Output}{Output}

    \caption{Abort for node $U_j \in V $}
        \label{algo:abort}

			   Upon receiving input $(U_i,abort)$\\
			   Set $flag=0$\\

			    $C(U_j,U_i)=C(U_j,U_i)+val_{j,i}$\\
			                \For{$id_{j,k} \in M_j$}
                {

                   \If{$isContract(id_{j,k})=1$}
                   {
                         $flag=1$\\
                   		\textbf{break from the loop}
                   }
			}
			\If{$flag=0$}
			{
			 \For{$U_m \in V: id_{m,j} \in E$}
                {
			
			Send $(abort)$ to $U_m$.
			}
			
			}

\end{proc}

\item Balance Security: Any intermediate node $U_i$ can claim payment from its preceding neighbors if at least one of the outgoing neighbors of $U_i$ accepted the payment. If $U_i$ receives abort from all the successors, it will abort as well. The total balance of $U_i$ either remains unchanged or it gains a processing fee $f(U_i)$.

\item Value Privacy: The ideal functionality $\mathcal{F}$ does not contact any user that does not 
belong to the set $\mathbb{PC}$, hence they learn nothing about the transacted value.

\item Unlinkability: For all the neighbors $U_j$ of node $U_i$, $\mathcal{F}$ samples a random identifier $h_{i,j}'$. Even if the neighbors collude, they cannot find any correlation amongst the payment identifiers.

\item Relationship Anonymity: Follows from unlinkability. If there exist at least one honest intermediate node $U_i$, then it receives a unique event identifier from $\mathcal{F}$ for each payment over any of its outgoing payment channels. Since all the event identifiers are independently generated, if at least one honest user $U_i$ lies in a payment path, any two simultaneous payments getting routed over the same set of payment channels for the same value $val$ is indistinguishable to the outgoing neighbors of $U_i$ receiving the request for forwarding the payments. This implies that any corrupted node cannot distinguish between the payments $(U_0,U_r,val)$ and $(U_0',U_r',val)$ with probability greater than $\frac{1}{2}$.

\item Atomicity: If $U_r$ triggers the release by responding with \emph{success}, it means that it has received all the partial payments. If $U_r$ fails to receive even one partial payment, then it will send \emph{abort} signaling a failed payment.

\end{itemize}

\subsection{Universal composability (UC) security}
The ideal functionality $\mathcal{F}$ can be \emph{attacked} by an ideal world adversary called a simulator or \emph{Sim}, a $\mathcal{PPT}$ algorithm. An additional special party called environment $\mathcal{Z}$ which observes both the real world and the ideal world, provides the inputs for all parties and receives their outputs. $\mathcal{Z}$ can use the information leaked by adversary $\mathcal{A}$ or actively influence the execution. Adversary $\mathcal{A}$ can corrupt any party before the protocol starts. However, the former doesn't get any information from communication occurring between honest parties. Let $REAL_{\Pi,\mathcal{A},\mathcal{Z}}$ be the ensemble of the outputs of the environment $\mathcal{Z}$ when interacting with the attacker $\mathcal{Z}$ and users running protocol $\Pi$,
\begin{definition}
\label{def1}
\textbf{UC Security}. Given that $\lambda$ is the security parameter, a protocol $\Pi$ \emph{UC}-realizes an ideal functionality $\mathcal{F}$ if for all computationally bounded adversary $\mathcal{A}$ attacking $\Pi$ there exist a probabilistic polynomial time ($\mathcal{PPT}$) simulator $Sim$ such that for all $\mathcal{PPT}$ environment $\mathcal{Z}$, $IDEAL_{\mathcal{F}, Sim,\mathcal{Z}}$, and $REAL_{\Pi,\mathcal{A},\mathcal{Z}}$ are computationally indistinguishable.

\end{definition}

\begin{figure}[!htb]
\begin{tcolorbox}[top=0pt,left=0.2pt,right=0.2pt,bottom=0pt,width=0.5\textwidth, colback=black!1,colframe=black!115!black,title={} ]

$isChannel(id_{i,j}):$
\begin{itemize}[leftmargin=*]
\item $\mathcal{F}$ sends $id_{i,j}$  to $\mathcal{F}_{\mathbb{B}}$. The latter checks for an entry in $\mathbb{B}$ of the form $(id_{i,j},v_{i,j},t_{i,j}',f_{i,j})$.
\item If the entry does not exist, then return 0.
\item If the entry exists, then check if there is an entry $id_{i,j}$ in $\mathcal{C}$. If it is true, then return 0, else return 1.
\end{itemize}

$isPred(sid,U_i,\mathbb{PC},V_{\mathbb{PC}}):$
\begin{itemize}[leftmargin=*]
\item For each  $U_k \in V_{\mathbb{PC}}: id_{k,i} \in \mathbb{PC}$:
\begin{itemize}[leftmargin=*]
\item If $contract(sid,id_{k,i})=0$, then return \emph{failure}.
\end{itemize}
\item Return \emph{success} 
\end{itemize}

\underline{PAY}

\paragraph*{(ii) \underline{Contract Forwarding Phase}}
 $U_0$ invokes $\mathcal{F}$ with message $(sid,pay,U_r,val,t_{end},\{(id_{i,j},val_{i,j},t_{i,j}):id_{i,j}\in \mathbb{PC}\},\mathbb{PC})$.

\begin{itemize}[leftmargin=*]

\item For each $id_{i,j} \in \mathbb{PC}$, set $contract(sid,id_{i,j})=0$. 

\item $\mathcal{F}$ forms a set $V_{\mathbb{PC}}=\{U_i\}$ such that $U_i \in V$ and has a channel in $\mathbb{PC}$.
\item Initialize an empty queue $Q_{pay}$. Push $U_0$ into queue $Q_{pay}$.
\item While $Q_{pay}$ is not empty:
\begin{itemize}[leftmargin=*]
\item Pop $U_i$ from $Q_{pay}$.

\item For each  $U_j \in V_{\mathbb{PC}}: id_{i,j} \in \mathbb{PC}$:
\begin{itemize}[leftmargin=*]
\item If $U_j$ sends $(sid,abort)$ to $\mathcal{F}$ then it removes all entries such entries from $\mathcal{L}$ added in this phase, cancel their contracts by resetting the variable to 0, and abort. 

\item $\mathcal{F}$ checks $isChannel(id_{i,j})=1$. If the check fails, then remove all entries $d_i$ from $\mathcal{L}$ added in this phase and abort.
\item Create $z_{i,j}=\{(id_{j,k},val_{j,k},t_{j,k}): \forall U_k \in V_{\mathbb{PC}}, id_{j,k} \in \mathbb{PC}\}$, if $U_j\neq U_r$. Else $z_{i,r}=\{val,t_{end}\}$.
\item $\mathcal{F}$ checks $ t_{i,j}\stackrel{?}{\geq} \max\limits_{U_k \in V_{\mathbb{PC}}, id_{j,k} \in z_{i,j}}\{t_{j,k}\}+\Delta$ and $val_{i,j}\leq\sum\limits_{U_k \in V_{\mathbb{PC}}, id_{j,k} \in z_{i,j}} val_{j,k}+f(U_j)$. If any of the checks fail, then remove all entries from $\mathcal{L}$ added in this phase, cancel their contracts by resetting the variable to 0, and abort.
\item $\mathcal{F}$ checks whether for $(id_{i,j},v_{i,j}',.,.)\in \mathcal{L}$, if $v_{i,j}'\geq val_{i,j}$. If that is the case, then add $d_{i,j}=(id_{i,j},v_{i,j}'-val_{i,j},t_{i,j},\bot)$ to $\mathcal{L}$, where $(id_{i,j},v_{i,j}',.,.) \in \mathcal{L}$ is the entry with the lowest $v_{i,j}'$. If the conditions are not met, $\mathcal{F}$ removes all entries from $\mathcal{L}$ added in this phase and abort. 
\item If the conditions are met, set $contract(sid,id_{i,j})=1$. Sample an identifier $h_{i,j}'$ and send request $(sid,\textrm{forward},U_i,id_{i,j},val_{i,j},t_{i,j},h_{i,j}',z_{i,j})$  to $U_j$. 
\item If $isPred(sid,U_j,\mathbb{PC},V_{\mathbb{PC}})$ returns $success$, push $U_j$ to $Q_{pay}$.
\end{itemize}
\end{itemize}
\end{itemize}

\end{tcolorbox}
\caption{Ideal World Functionality for payment in PCN}
\label{fpcn}
\end{figure}

\begin{figure}[!htb]

\ContinuedFloat 
\begin{tcolorbox}[top=0pt,left=0.2pt,right=0.2pt,bottom=0pt,width=0.5\textwidth, colback=black!1,colframe=black!115!black,title={} ]
%


\paragraph*{(ii) \underline{Release Phase}}
 $U_r$ invokes $\mathcal{F}$ with message $(sid,response)$.
\begin{itemize}[leftmargin=*]
\item For each $U_j \in V_{\mathbb{PC}}$:
\begin{itemize}[leftmargin=*]
\item Set $visited(U_j)=0$.
\end{itemize}
\item Initialize $flag_{abort}=0$ and initialize empty queues $\mathcal{T}$ and $Q_{failure}$.
\item If $response=\bot$, then set $flag_{abort}=1$.
\item If $flag_{abort}=0$, push $U_r$ in $T$.

\item While $\mathcal{T}$ is not empty:
\begin{itemize}[leftmargin=*]
\item Pop node $U_j$ from $\mathcal{T}$.
\item For each  $U_i \in V_{\mathbb{PC}}: id_{i,j} \in \mathbb{PC}$ and $contract(sid,id_{i,j})=1$:
\begin{itemize}[leftmargin=*]
\item Update $d_{i,j} \in \mathcal{L}$ to $(-,-,-,h_{i,j}')$, send $(sid,success,h_{i,j}')$ to $U_i$ and $U_j$, set $contract(sid,id_{i,j})=0$.
\item If $U_i$ sends $(sid,abort)$ then $visited(U_i)=1$, push $U_i$ in $Q_{failure}$.
\item Else if $visited(U_i)=0$ and $U_i\neq U_0$, set $visited(U_i)=1$ and push $U_i$ in $\mathcal{T}$.
\end{itemize}

\end{itemize}
\item If $flag_{abort}=1$, then :
\begin{itemize}[leftmargin=*]
\item Push $U_r$ to $\mathcal{T}$.

\item While $\mathcal{T}$ not null:
\begin{itemize}[leftmargin=*]
\item Pop node $U_j$ from $\mathcal{T}$. 
\item If $U_j\neq U_0$, go to the next step, else go back to previous step and continue. 
\item For each  $U_i \in V_{\mathbb{PC}}: id_{i,j} \in \mathbb{PC}$ and $contract(sid,id_{i,j})=1$:
\begin{itemize}[leftmargin=*]
\item set $contract(sid,id_{i,j})=0$. Remove $d_{i,j}$ from $\mathcal{L}$, send $(sid,\bot,h_{i,j}')$ to $U_i$ and $U_j$.
\end{itemize}
\item If $U_i \notin \mathcal{T}$, push $U_i$ in $\mathcal{T}$.
\end{itemize}

\end{itemize}
\item Else:
\begin{itemize}[leftmargin=*]
%
%
\item While $Q_{failure}$ is not empty:
\begin{itemize}[leftmargin=*]
\item Pop node $U_j$ from $Q_{success}$
\item For each  $U_i \in V_{\mathbb{PC}}: id_{i,j} \in \mathbb{PC}$ and $contract(sid,id_{i,j})=1$:
\begin{itemize}[leftmargin=*]
\item set $contract(sid,id_{i,j})=0$. Remove $d_{i,j}$ from $\mathcal{L}$, send $(sid,\bot,h_{i,j}')$ to $U_i$ and $U_j$.
\item If $visited(U_i)=0$, set $visited(U_i)=1$, push $U_i$ in $Q_{failure}$.
\end{itemize}

\end{itemize}

\end{itemize}
\end{itemize}
\end{tcolorbox}

\caption{Ideal World Functionality for payment in PCN (Continued)}
\end{figure}

\subsection{Security analysis}
\label{privan}
From Definition \autoref{def1}, a protocol $\Pi$ is said to be \emph{UC}-secure if $\mathcal{Z}$ cannot distinguish whether it is interacting with the ideal world or real-world even in presence of a computationally bounded adversary $\mathcal{A}$. Since our protocol execution in real-world relies on ideal functionalities $\mathcal{F}_{smt}$ and $\mathcal{F}_{\mathbb{B}}$, we define our protocol in the hybrid world \cite{fairswap} instead of real world. 

\begin{theorem}
\label{th2}
Given $\lambda$ is the security parameter, elliptic curve group of order $q$ is generated by the base point $G$, the protocol CryptoMaze \emph{UC}-realizes the ideal functionality $\mathcal{F}$ in the $(\mathcal{F}_{\mathbb{B}}, \mathcal{F}_{smt}$)-hybrid world. 
\end{theorem}
\begin{proof} We design \emph{Sim} for the ideal world execution for the following cases: either the sender is corrupt or the receiver is corrupt, or one of the intermediate node is corrupt. The only event which distinguishes hybrid world from ideal world is when the \emph{Sim} aborts in ideal world. 
\begin{itemize}[leftmargin=*]
\item \underline{$U_0$ is corrupted}: $\mathcal{A}$ acts like the sender $U_0$, and forms packet $(R_{i,j},val_{i,j},t_{i,j},Z_{i,j})$, for each $id_{i,j}\in \mathbb{PC},U_i \neq U_r$. The encrypted message $Z_{i,j}$ upon decryption gives $M_{j}=\{(id_{j,k},val_{j,k},x_{j,k},R_{j,k},t_{j,k},Z_{j,k}): \forall k \in V, id_{j,k} \in \mathbb{PC}\}$, when $U_j \neq U_r$ and $M_j=\{(y_{i,j},t_{end})\}$, when $U_j=U_r$. $\mathcal{A}$ forwards the packet to \emph{Sim}. 

For each node $U_i \in V, U_i \neq \{U_0,U_r\}$, \emph{Sim} does the following:
\begin{itemize}[leftmargin=*]
\item Form set $D_{i}=\{(id_{i,k},x_{i,k},R_{i,k},t_{i,k}):\forall k \in V, id_{i,k} \in \mathbb{PC}\}$.
\item For each $U_j \in V: id_{j,i} \in \mathbb{PC}$:
\begin{itemize}[leftmargin=*]
\item Get $(R_{j,i},t_{j,i})$, call \textbf{TimeLockContractForward} with input $(D_i,t_{j,i},R_{j,i})$ as input. If it returns failure, then abort.
\end{itemize}
\emph{Sim} checks $\sum\limits_{U_j \in V: id_{j,i} \in \mathbb{PC}} val_{j,i} = \sum\limits_{U_k \in V: id_{i,k} \in \mathbb{PC}} val_{i,k}+f(U_i)$. If the check fails, abort.
\end{itemize}

If the process didn't abort, \emph{Sim} sends $(sid,pay,U_r,val,t_{end},\{(id_{i,j},val_{i,j},t_{i,j}): id_{i,j} \in \mathbb{PC}\}, \mathbb{PC})$ to $\mathcal{F}$. \emph{Sim} has already checked the flow consistency for the intermediate honest nodes, as well as consistency of terms of incoming and outgoing contracts before it forwards the conditional payment.

In the \emph{release phase}, if $U_r$ aborts, then the process aborts as well. If $U_r$ has released the secret, then \emph{Sim} checks that any node $U_i$ claiming payment from $U_j$ has released the discrete logarithm for $R_{j,i}$. We consider that an honest intermediate node $U_m$ splits the transaction value across multiple payment channels, and a partial value gets routed via channel $id_{m,i}$. We identify a bad event $E_1$: if adversary $\mathcal{A}$ has released $r_{m,i}$ for $R_{m,i}: r_{m,i}G=R_{m,i}$ but $\exists U_k$ where $id_{k,m} \in \mathbb{PC}$, $r=r_{m,i}+e_{k,m} x_{m}+x_{m,i}$ and $R_{k,m}\neq rG$ then \emph{Sim} aborts the simulation.  
\begin{claim}
\label{cl1}
The probability of $E_1$ is 0. 
\end{claim}
\textit{Proof.} \emph{Sim} checks the relation $R_{k,m}\stackrel{?}{=}R_{m,i}+e_{k,m} x_{m}G+x_{m,i}G, \forall U_i \in V, id_{m,i} \in \mathbb{PC}$ at the start. If $R_{m,i}=r_{m,i}G$ but $R_{k,m}\neq r G$, then $r\neq r_{m,i}+e_{k,m} x_{m}+x_{m,i}$, which contradicts event $E_1$. Hence, the probability is 0.

\item \underline{An intermediate party $U_m$ is corrupted}: If $\mathcal{A}$ can release the discrete logarithm of the statement used in the incoming channel's contract of node $U_m$ before the secret is revealed, \emph{Sim} aborts. 

When \emph{Sim} gets $(sid,\textrm{forward},U_j,id_{j,m},val_{j,m},t_{j,m},h_{j,m}',\\z_{j,m})$ from $\mathcal{F}$ on behalf of all incoming nodes $U_j$ of node $U_m$, it samples $x_{m,k}$ for each $U_k \in z_{j,m}$, computes $x_m=\sum\limits_{U_k \in V, id_{m,k} \in \mathbb{PC}} x_{m,k}$. \emph{Sim} sends $(forward,R_{j,m},val_{j,m},t_{j,m})$, $Z_{j,m}$ to $\mathcal{A}$ on behalf of all $U_j$s. $\mathcal{A}$ sends $(R_{m,k},val_{m,k},t_{m,k})$ to \emph{Sim} for all such $U_k$s, on behalf of $U_m$. \emph{Sim} checks whether  $\forall U_j \in V$, $R_{j,m}\stackrel{?}{=}x_mG+R_{m,k}+e_{j,m} x_{m,k}G$ and $t_{j,m}\stackrel{?}{=} \Delta+ \max\limits_{U_k \in V, id_{m,k} \in \mathbb{PC}} \{t_{m,k}\}$ and $\sum\limits_{U_j \in V, id_{j,m} \in \mathbb{PC}} val_{j,m}\stackrel{?}{=}\sum\limits_{U_k \in V, id_{m,k} \in \mathbb{PC}} val_{m,k}+f(U_m), \forall U_k \in V,id_{m,k}\in \mathbb{PC}$. If any of the checks fail, then it sends abort to $\mathcal{F}$.

Consider $U_t$ as the incoming node forwarding payment to $U_m$ and $U_h$ as the outgoing neighbor of $U_m$. \emph{Sim} samples $r^*$ such that $R_{m,h}=r^*G$ and $R_{t,m}=x_{m,h}G+e_{t,m} x_{m}G+R_{m,h}$. We identify another bad event $E_2$: if $\mathcal{A}$ releases $r'$ such that $R_{t,m}=r'G$ without querying \emph{Sim} on the event identifier $h_{m,h}'$, \emph{Sim} aborts. 
\begin{claim}
\label{cl2}
The probability of $E_2$ is $\frac{1}{q}$, where $\mathbb{G}$ is an elliptic curve group with large order $q$ i.e. $|\mathbb{G}|=q$.
\end{claim}
\textit{Proof.} Follows from the discrete logarithm hardness assumption, given a random point $h \in \mathbb{G}$, it is possible to guess the value $log_{G}h$ with probability $\frac{1}{q}$.

\item \underline{$U_r$ is corrupted}: \emph{Sim} receives $(sid,\textrm{forward},U_j,id_{j,r},\\ val_{j,r},t_{j,r},h_{j,r}',z_{j,r})$ on behalf of all incoming nodes $U_j$ of node $U_r$ from $\mathcal{F}$. \emph{Sim} gets $X_{\acute{r}}$ from $\mathcal{A}$ and samples $y_j$, creates $R_{j,r}= X_{\acute{r}}+ e_{j,r} yG$, where $y=\sum\limits_{U_j \in V: id_{j,r} \in \mathbb{PC}} y_j$, for all the incoming neighbors $U_j$ of node $U_r$.  It sends $(forward,R_{j,r},val_{j,r},t_{j,r})$, $Z_{j,r}$ to $\mathcal{A}$ on behalf of all $U_j$s. We identify another bad event $E_3$: if there exists a node $U_k: id_{k,r}\in \mathbb{PC}$ such that $\mathcal{A}$ releases $x'$ such that $R_{k,r}=x'G$ without querying \emph{Sim} on the event identifier $h_{k,r}'$, \emph{Sim} aborts the simulation.
\begin{claim}
\label{cl3}
The probability of $E_3$ is $\frac{1}{|q|}$, where $\mathbb{G}$ is an elliptic curve group with large order $q$ i.e. $|\mathbb{G}|=q$.
\end{claim}
\textit{Proof.} $\mathcal{A}$ knows $\textrm{dlog}(X_{\acute{r}})$, but it doesn't know $y$. Hence, $\mathcal{A}$ can guess $\textrm{dlog}(R_{k,r})$ with probability $\frac{1}{q}$.
\end{itemize}

\end{proof}
\textit{Indistinguishability from the ideal world}. The simulator \emph{Sim} designed is efficient since it runs a polynomially-bounded algorithm. To argue that $\mathcal{Z}$'s view in simulation is indistinguishable from the execution protocol in the hybrid-world protocol, we consider the occurrence of a bad event in \texttt{PAY}:\\
(i) When $U_0$ is corrupted, the random values sampled by \emph{Sim} and the values are chosen by an honest $U_0$ follow the same distribution. Similarly, when $U_r$ is corrupted or an intermediate node is corrupted, the random values sampled by \emph{Sim} remain indistinguishable from the data used in honest execution. \\
(ii) Indistinguishability breaks when \emph{Sim} aborts in the ideal world. We infer from Claim \ref{cl1}, Claim \autoref{cl2}, and Claim \ref{cl3}, that bad events occur with negligible probability and hence \emph{Sim} aborts with negligible probability.  

Thus, we have proved that our protocol CryptoMaze \emph{UC}-realizes the ideal functionality $\mathcal{F}$ in the $(\mathcal{F}_{\mathbb{B}},\mathcal{F}_{smt})$-hybrid world. If the security and privacy goals stated in Section \ref{pg} are realized by $\mathcal{F}$, then as per \emph{UC Definition of Security} stated in Definition \ref{def1} these security notions are satisfied by our protocol as well.

\section{Experimental analysis}
\label{exp}

We choose to compare our protocol with \emph{Multi-Hop HTLC} \cite{malavolta}, \emph{Atomic Multi-path Payment} \cite{multipath} and \emph{Eckey et al.} \cite{eckey2020splitting}. \emph{Multi-Hop HTLC} is a single-path payment protocol, and we show how extending it to a multiple-path payment would work. In this protocol, a node forwarding payment in a given path gets a tuple $(y,h_1,h_2)$ along with the non-interactive zero-knowledge proof $\Pi$ for the statement $``\exists x': h_1=\mathcal{H}(x')  \ and \ h_2=\mathcal{H}(y \oplus x')''$. \emph{Atomic Multi-path Payment} or \emph{AMP} is the most efficient protocol in the existing state-of-the-art in terms of run time (considering best-case run time) as well as communication cost. If there are $n$ paths in \emph{AMP}, then the payer generates secret shares $x_1,x_2,\ldots,x_n$ for each path from the master secret $x: x=x_1 \oplus x_2 \oplus \ldots \oplus x_n$. We find the objective of the protocol proposed in \emph{Eckey et al.} similar to ours. However, the payment split is decided on the fly and sharing of the public key leads to linkability between partial payments. We explain briefly the state-of-the-art with which we have compared our protocol.

\subsubsection{Multi-Hop Hashed Timelock Contract} 

We discuss a protocol Multi-Hop HTLC with an example. The protocol  preserves privacy of payment and hides the identity of payer and payee \cite{malavolta}. 

\textit{Construction.} In the Fig. \ref{mhtlc}, Alice samples 4 random numbers $x_1,x_2,x_3$ and $x_4$. It constructs $y_4=\mathcal{H}(x_4)$ where $\mathcal{H}$ is any standard one-way hash function. Next, it constructs $y_3=\mathcal{H}(x_3 \oplus x_4)$ and a zero-knowledge proof $\pi_3$ for the statement \emph{``given $y_3$ and $y_4$, there exists an $x: y_4=\mathcal{H}(x)$ and $y_3=\mathcal{H}(x_3\oplus x)$''}. Similarly, it constructs $y_3=\mathcal{H}(x_2 \oplus x_3 \oplus x_4)$ and a zero-knowledge proof $\pi_2$ for the statement \emph{``given $y_2$ and $y_3$, there exists an $x: y_3=\mathcal{H}(x)$ and $y_2=\mathcal{H}(x_2\oplus x)$''}. It constructs $y_1=\mathcal{H}(x_1 \oplus x_2 \oplus x_3 \oplus x_4)$ and a zero-knowledge proof $\pi_1$ for the statement \emph{``given $y_1$ and $y_2$, there exists an $x: y_2=\mathcal{H}(x)$ and $y_1=\mathcal{H}(x_1\oplus x)$''}. It sends the value $(x_1,y_1,y_2,\pi_1)$ to Bob, $(x_2,y_2,y_3,\pi_2)$ to Charlie, $(x_3,y_3,y_4,\pi_3)$ to Eve and $(x_4,y_4)$ to Dave via a secure anonymous channel. 

\emph{Contract Creation Phase.} Bob, Charlie, and Eve check whether the zero-knowledge proof received is correct or not. Alice forms the contract with Bob using condition $y_1$. If proof $\pi_1$ is correct, Bob accepts the payment, else he will abort. Bob forwards the payment to Charlie using the condition $y_2$. Charlie forwards the payment to Eve using the condition $y_3$ and Eve does the same to Dave using the condition $y_4$.

\emph{Contract Release Phase.} Upon receiving the conditional payment, Dave checks if $y_4\stackrel{?}{=} \mathcal{H}(x_4)$. If this holds true, Dave sends $x_4$ to Eve and claims payment. Eve calculates $x_3 \oplus x_4$, sends it to Charlie and claims payment. Charlie computes $x_2 \oplus x_3 \oplus x_4$, claims payment from Bob upon releasing this key. Bob computes $x_1 \oplus x_2 \oplus x_3 \oplus x_4$ and claims payment from Alice.
\begin{figure}
    \centering
    \includegraphics[scale=0.35]{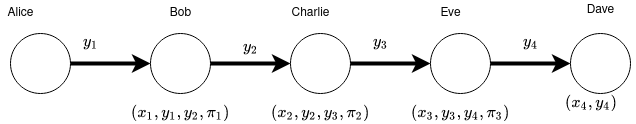}
    \caption{Multi-Hop HTLC construction for payment from Alice to Dave}
    \label{mhtlc}
    \end{figure}
None of the intermediate participants can correlate the payments as every channel uses a different condition. However, Multi-Hop HTLC requires exchanging non-trivial amount of data and computation of complex zero-knowledge proof during setup phase.

\subsubsection{Atomic Multi-path Payment}

Atomic MultiPath Payment \cite{multipath} splits payment across multiple path, guaranteeing atomicity. Once the receiver receives all the conditional payments from different routes, only then it can claim the payment. The payment hash used across different routes is different, preventing any correlation. The setup is non-interactive, where the payer need not coordinate with the payee. We explain the protocol with an example.

\textit{Construction.} Alice needs to send an amount $v$ to Bob. It figures $n$ paths $P_1,P_2,\ldots,P_n$, with each path transferring $v_1,v_2,\ldots,v_n: v=\sum\limits_{i=1}^n v_i$ as shown in Fig. \ref{amp}. Alice samples secret $s_1,s_2,\ldots,s_n$. The master secret $s=s_1\oplus s_2\oplus \ldots \oplus s_n$ can be generated. Using $s$, she generates the condition of payment for each path $P_i$ as follows: $H_i=\mathcal{H}(s||i), i\in [1,n]$. For each path $P_i$, the conditions of payment $H_i$ is forwarded using onion routing, where the tuple $(s_i,i)$ is sent as an encrypted onion blob or EOB which can only be decrypted by Bob. Upon receiving all the conditions from n paths, Bob computes $s$ and constructs the preimage $s||i$ for each path $P_i$ to claim payment. If any of the paths fails, then Bob will not be able to claim payment. 
\begin{figure}
    \centering
    \includegraphics[scale=0.3]{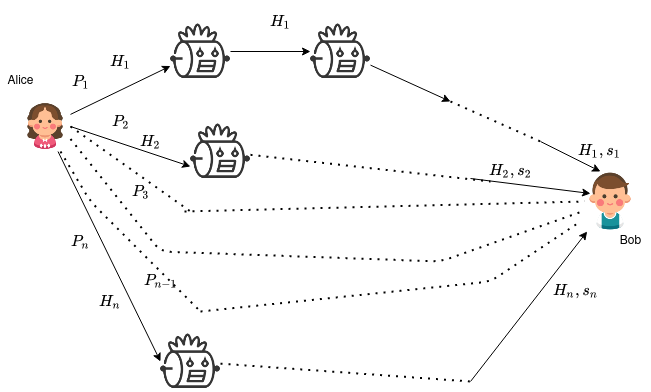}
    \caption{Atomic MultiPath Payment from Alice to Bob}
    \label{amp}
    \end{figure}

\subsubsection{Splitting Payments Locally While Routing Interdimensionally}

Eckey et al. \cite{eckey2020splitting} had proposed an atomic payment protocol that allows intermediaries to split payments dynamically by adapting to the local condition. Instead of the payer specifying the path, each intermediate party independently decides on the split of the payment value. Initially, the receiver sends the hash of a secret preimage $x_R$, denoted as $H_R=\mathcal{H}_{ah}(x_R)$, to the sender, the hash function used here is an additive homomorphic one-way function. The latter, upon checking its neighbor and their residual capacity, decides on the conditional payment by sampling different random values for each split. If the sender splits the payment amount into $k$ parts, then it samples $x_1, x_2, \ldots, x_k$. The hash of the receiver's preimage and the hash of the sender's preimage for each split is added, $H+\mathcal{H}(x_i), i \in [1,k]$ and forwarded to the neighbor along with the address of the receiver and encrypted value of each random value, denoted as $Enc_{HE}(x_i)$. The encryption used here is homomorphic in nature, using the public key of the receiver. If two encrypted values are added, then upon decryption, we get the summation of these two values. The neighbor upon receiving the packet decides upon the next neighbor which can forward the packet to the receiver. It performs the same step as done by the sender and combines its encrypted random value with that received from the sender. In the end, when the receiver receives the conditional payments for all the splits, it decrypts the encrypted value and adds the preimage it had sampled initially. It now claims the payment by releasing this preimage. Sender can generate a valid receipt of the payment, provided the sender receives the secret preimage sampled by the receiver. For routing, it uses a new algorithm, \emph{Interdimensional SpeedyMurmur}.

\subsection{Evaluation methodology}
In this section, we define the experimental setup. The code is available in \cite{Code}. System configuration used is \texttt{Intel Core i5-8250U CPU}, Operating System: \textit{Kubuntu-20.04.1}, and Memory: 7.7 GiB of RAM. The programming language used is C, compiler - gcc version 5.4.0 20160609. For implementing the cryptographic primitives in  CryptoMaze, Atomic Multi-path Payment or \emph{AMP} and \emph{Multi-Hop HTLC}, we use the library \textit{OpenSSL}, version-1.0.2 \cite{openssl}. For constructing the zero-knowledge proof for \emph{Multi-Hop HTLC}, we have used C-based implementation of ZKBoo\cite{Zkboo} and libgcrypt version-1.8.4 \footnote{\url{https://gnupg.org/software/libgcrypt/index.html}}. The number of rounds for ZKBoo is set to 136. This guarantees a soundness error of $2^{-80}$ for the proof and witness length is set to 32 bytes. For elliptic curve operations in CryptoMaze and \emph{Eckey et al.}, we have considered the elliptic curve secp224r1. For homomorphic encryption using \emph{Paillier Cryptosystem} in \emph{Eckey et al.}, \emph{libhcs} is used \cite{libhcs}. It is a C library implementing several partially homomorphic encryption schemes \cite{damgaard2010generalization}. 


\subsubsection{Metric used}
The following metrics are used to compare the performance of \textit{CryptoMaze} with other state-of-the-art protocols.
\begin{itemize}[leftmargin=*]
\item TTP (\textit{Time taken for payment}): It is the time taken for searching of eligible paths for routing a payment, formation of off-chain payment contracts, and completion of payment upon successfully fulfilling the criteria set in the contract. It is measured in seconds or \emph{s}.

\item Communication Overhead: For the given payment protocol, the number of messages exchanged between the nodes while searching for a set of paths and execution of the payment protocol, measured in kilobytes or \emph{KB}.
\end{itemize}

\begin{figure*}[ht]

    \centering
    \begin{subfigure}[b]{0.32\textwidth}
    \centering
    {\includegraphics[height=1.3in]{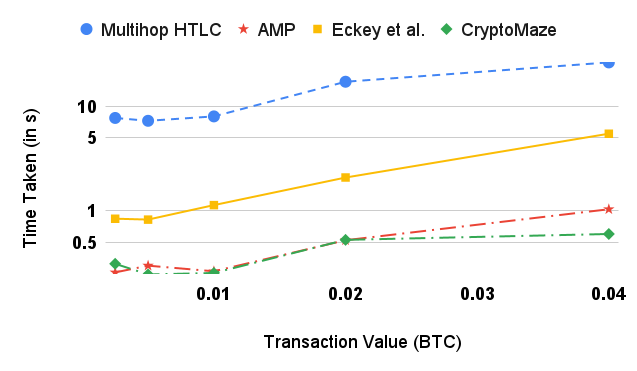}  }
    {\includegraphics[height=1.3in]{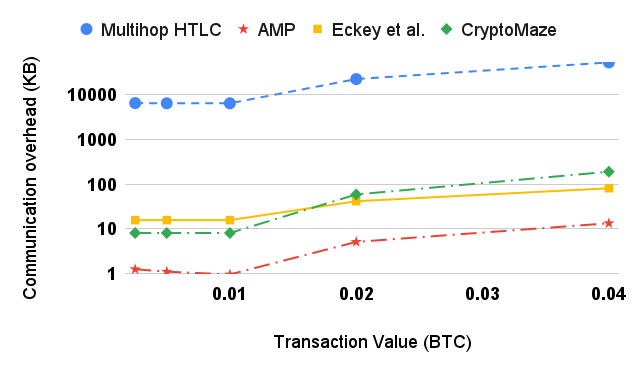} }   
    \label{march20}    
     \caption{LN snapshot March 2020}
\end{subfigure}
    \begin{subfigure}[b]{0.32\textwidth}
    \centering
    {\includegraphics[height=1.3in]{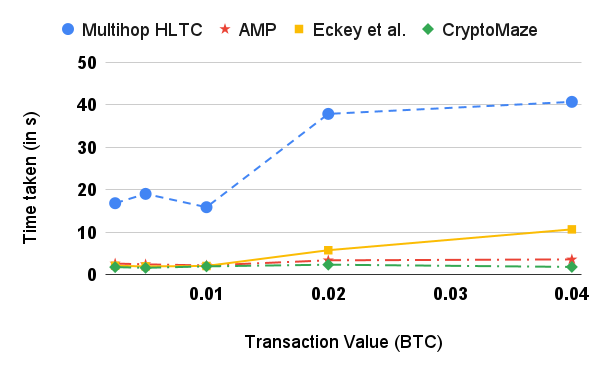} }
    {\includegraphics[height=1.3in]{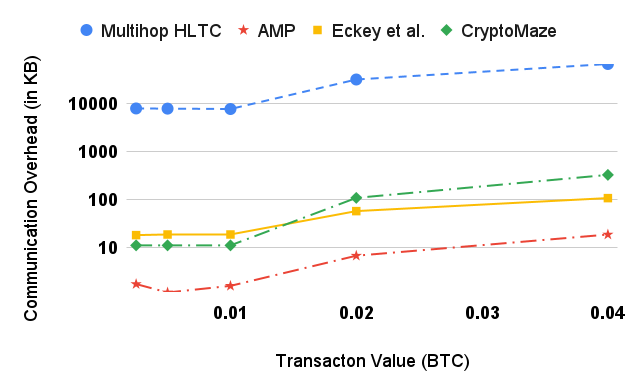}  }
    \label{april21}
     \caption{LN snapshot May 2021}
\end{subfigure}
    \begin{subfigure}[b]{0.32\textwidth}
    \centering
{    \includegraphics[height=1.3in]{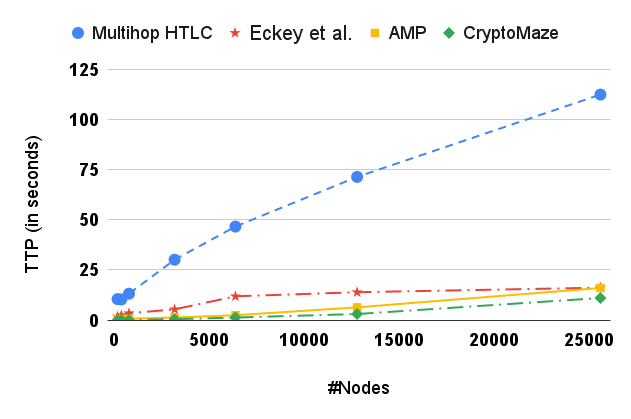}  }
{    \includegraphics[height=1.3in]{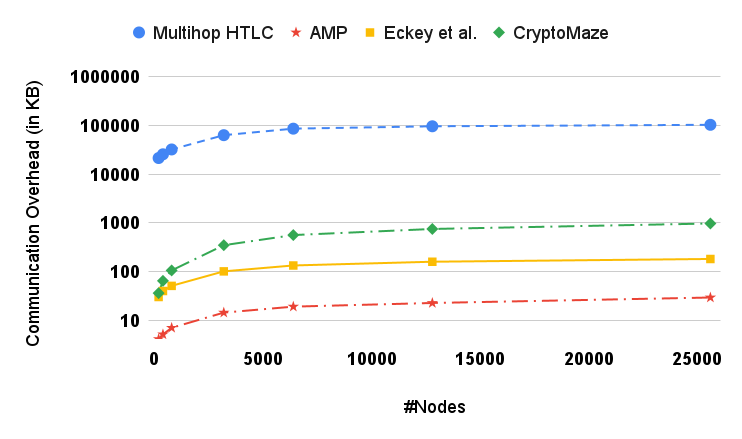}}

     \caption{Simulated Network}
\end{subfigure}
\caption{Experimental Analysis: Time taken for payment and Communication overhead}
\label{simulated}    
\end{figure*}

\subsection{Observations}
We use the distributed routing algorithm \emph{HushRelay} \cite{mazumdar2020hushrelay} for our protocol, \emph{Atomic Multi-path Payment} and \emph{Multi-Hop HTLC}, that returns the set of paths. Based on this set of paths as input, we run each instance of the payment protocol. Since the time taken for routing is taken into account while estimating TTP, it can be further optimized by using a more efficient distributed routing algorithm. 
\subsubsection{Evaluation on real instances}
We select two snapshots of Lightning Network taken on March 2020 \cite{Ayelet} and May 2021\footnote{\url{https://www.dropbox.com/s/fkq7kh5xyu3l33t/LN_25_05_2021.json?dl=0}}. The first instance has \emph{6329 nodes} and the second instance has \emph{11072 nodes}. The payment amount is varied between 0.0025 BTC - 0.04 BTC.

(a) \textit{Optimization in terms of off-chain contracts.} Before stating the observation in terms of execution time and communication cost, we analyze the saving in terms of off-chain contracts established in shared edges. Since state-of-the-art protocols instantiate multiple contracts on shared edges, we choose anyone as a representative and compare it with our protocol. We run 20000 payment instances for a given transaction value and state our results for the two Lightning Network instances.
\begin{itemize}[leftmargin=*]
\item \emph{LN instance, March 2020}: When the transaction value was increased from 0.0025 BTC to 0.04 BTC, the payment instances that had multiple routes sharing payment channels increased from $1\%$ to $38\%$. The number of payment channels shared in a single payment instance increased from $33\%$ to $55\%$. The number of times a particular payment channel got shared increased from $2$ to $5$. The total number of off-chain contracts per payment instance increased from $24\%$ to $68.75\%$ for state-of-the-art.   
\item \emph{LN instance, May 2021}: Payment instances sharing channels for a single payment increased from $0.04\%$ to $38.6\%$. Channels shared for a given instance increased from $33\%$ to $42.8\%$. A channel gets shared not more than $4$ times. The total number of off-chain contracts per payment instance increased up to $54.5\%$ for state-of-the-art.    

\end{itemize}

(b) \textit{Computation and Communication Cost.} We analyze the efficiency of CryptoMaze compared to state-of-the-art in terms of the metric stated, when executed on a single payment instance.
\begin{itemize}[leftmargin=*]
\item  TTP for CryptoMaze is equivalent to \emph{Atomic Multi-path Payment}, not exceeding 0.39s on average as shown in \autoref{simulated} (a) and it is around 1.85s in \autoref{simulated} (b) for the second snapshot. Our protocol is approximately 3 times faster than \emph{Eckey et al.} and 17.5 times faster than \emph{Multi-Hop HTLC} for both instances.
\item  The communication overhead in \autoref{simulated} (a) is 53.18KB and in \autoref{simulated} (b) is 93.203KB, on average. The overhead is 14.5 times greater than that of \emph{Atomic Multi-path Payment} and 2 times more than that of \emph{Eckey et al.} for both instances. The communication overhead of \emph{Multi-Hop HTLC} is 297 times more than CryptoMaze.

\end{itemize} 

\subsubsection{Evaluation on simulated instances} 
Payment channel networks follow a small-world, scale-free structure \cite{rohrer2019discharged}. For generating synthetic graphs of size ranging from 200 to 25600 based on Bar\'{a}basi-Albert model \cite{albert2002statistical}, \cite{barabasi2003scale}, library \textit{igraph} was used. Optimization in terms of off-chain contracts is not analyzed since these are synthetic graphs. The topology of the synthetic graph may not be able to mimic the execution of multiple payment instances in the Lightning Network. We make the following observations based on executing a single payment instance:
\begin{itemize}[leftmargin=*]
\item  TTP for CryptoMaze increases gradually with the increase in the size of the network. The execution time does not exceed 11s upon execution on an instance of size 25600. Run time of \emph{AMP} is 1.7 times of CryptoMaze, that of \emph{Eckey et al.} and \emph{Multi-Hop HTLC} is 3.5 times and 18 times that of our protocol on an average. The plot is given in \autoref{simulated}(c). 

\item  The communication overhead in \autoref{simulated} (c) increases with an increase in the size of the network, with the communication overhead not exceeding 1000KB or 1MB on an instance of size 25600. On average, the communication overhead of CryptoMaze is 5 times of \emph{Eckey et al.} and 33 times of \emph{Atomic Multi-path Payment}. However, the overhead is 105 times less compared to \emph{Multi-Hop HTLC}.

\end{itemize}

\subsection{Discussion}
(i) \emph{Optimization in terms of off-chain contracts}: When the transaction amount per payment was increased, the liquidity of channels decreased. Payments were split into smaller amounts and routed via multiple paths. Thus, we observed that the number of instances where the routes were not edge-disjoint increased. With the increase in transaction amount, the number of paths routing a payment increased due to the increase in the split. The number of off-chain contracts established per payment increased for the state-of-the-art protocols. When the size of the network increases, the higher the chance of finding routes with higher capacity, the more options of edge-disjoint routes. Hence, a decrease in the number of off-chain contracts is observed. 

CryptoMaze combines the conditions for each of the partial payments routed via shared edges and form a single off-chain contract, our protocol saves around $50\%-60\%$ compared to state-of-the-art in terms of setup cost. Also, it does not have to pay a node more than once for routing payment, thus saving on the processing fee.   

(ii) \emph{Efficiency in terms of computation and communication cost}: We discuss our observation in terms of the metric used.
\begin{itemize}[leftmargin=*]
\item Time taken to execute CryptoMaze is comparable to \emph{AMP}, sometimes even lower than the latter. The reason is the mapping set of routes into a set of edges before establishing the off-chain contracts. All the previous protocols considered each route individually, increasing the setup time. \emph{Eckey et al.} have a higher run time due to the use of homomorphic encryption. In \emph{Multi-Hop HTLC}, generating zero-knowledge proofs for the preimage of a given hash value is an expensive process in terms of computation cost.

Overall, the time taken to execute the payment protocol increases slightly with an increase in the transaction amount and an increase in the network size. The higher is the transaction amount, the higher the chance of the payment being split into multiple partial payments. When the network size increases, the time taken to process the network for searching paths for routing the payment increases as well.

\item It is observed that \emph{AMP} has the lowest communication overhead because each node forwards just a single commitment to its neighbor in the path routing payment. However, each path is susceptible to wormhole attack. \emph{Eckey et al.} have a higher communication overhead. Here, each node forwards the public key and an encrypted message to its neighbor. In CryptoMaze, each node forwards a set of conditions and a set of secret values to its neighbor. The communication overhead is slightly greater than \emph{Eckey et al.}. However, the surge in communication overhead is to some extent compensated in the shared channels, where a single off-chain contract instead of multiple off-chain contracts. \emph{Multi-Hop HTLC} has the highest communication cost. The zero-knowledge proof $\Pi$ forwarded to each node has a significant size, plus multiple off-chain contracts are formed on shared edges, increasing the communication overhead.

\end{itemize}

The result demonstrates that our proposed protocol is efficient and scalable in
terms of computation cost and resource utilization.

 \begin{figure*}[!htb]
\begin{center}
\fbox{
\begin{minipage}{0.8\textwidth} 
\underline{\textbf{KeyGen}}\\
Upon receiving $(sid,\textrm{keygen}, U_j)$ from $U_i$ and
$( sid, \textrm{keygen},U_i)$ from $U_j$:
\begin{itemize}[leftmargin=*]
\item Sample a secret key $sk \leftarrow \mathbb{Z}_q$
\item Compute a public key $pk=sk.\mathcal{G}$
\item Output the message $(\textrm{keygen},sid, pk)$ to $U_i$ and $U_j$
\item Store $(sid,\textrm{keygen},sk)$
\end{itemize}
\vspace{0.3cm}
\underline{\textbf{Lock}}\\
Upon receiving $(sid,lock,m,R_{i,j},pk)$ from both $U_i$ and $U_j$:
\begin{itemize}[leftmargin=*]
\item If $(sid,lock)$ is already stored, abort.
\item Check if $( sid, \textrm{keygen}, sk)$ for the given $pk : pk =sk\mathcal{G}$ has been stored.
\item Sample $k \leftarrow \mathbb{Z}_q$ and compute $(r_x, r_y) = R = kR_{i,j}$
\item Query the Random Oracle at point $(sid,m)$, which returns $\mathcal{H}(m)$.
\item Compute $s = k^{-1}(\mathcal{H}(m) + r_x · sk)$
\item Send a output $(lock, sid, (r_x, s))$ to $U_i$ and $U_j$
\item Store $( sid,lock)$
\end{itemize}
\vspace{0.3cm}
\underline{\textbf{Verify}}\\
Upon receiving $(sid,verify,m,r',z',pk)$ from both $U_i$ and $U_j$, where $R_{i,j}=r_{i,j}\mathcal{G}$:
\begin{itemize}[leftmargin=*]
\item If $( sid, lock)$ is not stored then abort.
\item Parse $z'$ and retrieve $(r_x,s)$
\item Query the Random Oracle at point $(sid,m)$, which returns $\mathcal{H}(m)$.
\item Compute $s' = \frac{s}{r_{i,j}}$ and $(s_x,s_y)=S'=\frac{\mathcal{H}(m)\mathcal{G}+r_x.pk}{s'}$\\
\item Check $s_x\stackrel{?}{=}r_x$, if true return $(sid,verified)$ to $U_j$

\end{itemize}

\end{minipage}
}
%
\end{center}
\caption{Interface of ideal world functionality $\mathcal{F}_{ECDSA-Lock}$}
 \label{define}
\end{figure*}

\section{Use of scriptless lock in Cryptomaze}
\label{lock}
We leverage the use of scriptless scripts, where a signature scheme can be used simultaneously for authorization and locking. The crux of a scriptless locking mechanism is that the lock can consist only of a message $m$ and a public key $pk$ of a given signature scheme, and can be released only with a valid signature $\sigma$ of $m$ under $pk$. We next define how scriptless ECDSA signature can be used as a locking mechanism, the construction is similar to the one defined in \cite{malavoltamulti}. %
The main idea used here is that the locking algorithm is initiated by two users $U_{i}$
and $U_j$ who agree on a message $m$, for our purpose we consider $m=id_{i,j}$, and on the value $R_{i,j} = r_{i,j} G$ of the unknown discrete logarithm. The two parties then generate a random number $k$ and agree on a randomness $R = k R_{i,j}$. The shared \emph{ECDSA signature}
is computed by \emph{``ignoring"} the $R_{i,j}$, since the parties
are unaware of its discrete logarithm. The signature computed is $(r_x,s)$ where it can be written as $(r_x,s'r_{i,j})$. The signature $(r_x,s')$ is a valid \emph{ECDSA signature} on $m$. Once $r_{i,j}$ is released by node $U_j$, it is used for completing the signature.

We define this as an ideal functionality $\mathcal{F}_{ECDSA-Lock}$ in Fig. \ref{define}, which has access to a Random Oracle. The interfaces are \textbf{KeyGen, Lock }and \textbf{Verify}. KeyGen generates a common public key for a payment channel $id_{i,j}$ between parties $U_i$ and $U_j$. The Lock Phase and Verify Phase have been discussed previously. CryptoMaze accesses this ideal functionality $\mathcal{F}_{ECDSA-Lock}$ for forming the lock and releasing it as well.

\section{Conclusion}
\label{7}
In this paper, we propose a novel privacy-preserving, atomic multi-path payment protocol CryptoMaze. Multiple paths routing partial payments are mapped into a set of edges. Off-chain contracts are instantiated on these edges in a breadth-first fashion, starting from the sender. The use of this technique avoids the formation of multiple off-chain contracts on channels shared across multiple paths, routing partial payments. Partial payments remain unlinkable that prevents colluding parties from censoring split payments. We analyze the performance of the protocol on some instances of Lightning Network and simulated networks. From the results, we infer that our protocol has less execution time and feasible communication overhead compared to existing payment protocols.

As part of our future work, we intend to improve the protocol by incorporating a dynamic split of payments, similar to the work in \cite{eckey2020splitting}. This will reduce the computation overhead of the sender by eliminating the preprocessing step of constructing conditions for each off-chain contract. However, the main challenge is to realize such a protocol without violating unlinkability. 

\section*{Acknowledgment}
We thank Dr. Sandip Chakraborty, Associate Professor in the Department of Computer Science and Engineering at the Indian Institute of Technology (IIT) Kharagpur, for his initial valuable comments on this work. We also thank the anonymous reviewers and associate editor of IEEE TDSC for their comments.

\bibliographystyle{ieeetr}

\bibliography{PCN}

\end{document}